\documentclass[11pt]{article}
\usepackage{multicol}

\newtheorem{theorem}{Theorem}

\newtheorem{proof}{Proof}
\usepackage{graphicx}% Include figure files
\usepackage{dcolumn}% Align table columns on decimal point
\usepackage{bbm}
\usepackage{bm}
\usepackage{geometry}
\usepackage{cite}
\usepackage{braket}
\usepackage{lipsum}
\usepackage{mathtools}
\usepackage{lipsum}
\usepackage{bbold}
\usepackage{subcaption}
\usepackage{dsfont}

\usepackage[utf8]{inputenc}
\usepackage[T1]{fontenc}

\usepackage{xcolor}
\usepackage[colorlinks=true,allcolors=blue]{hyperref}
\newcommand*{\SavedEqref}{}
\let\SavedEqref\eqref
\renewcommand*{\eqref}[1]{%
	\begingroup
	\hypersetup{
		linkcolor=linkequation,
		linkbordercolor=linkequation,
	}%
	\SavedEqref{#1}%
	\endgroup
}
\DeclareMathOperator{\Tr}{Tr}
\usepackage{algorithm}
\usepackage{algpseudocode}

\begin{document}
	\title{Probabilistic Pathways: New Frontiers in Quantum Ensemble Control}
	
	\author{ Randa Herzallah\thanks{\textcolor{blue}{randa.herzallah@warwick.ac.uk}}\\Warwick Mathematics Institute,
		The University of Warwick, Coventry, CV4 7AL, UK.\\
		\\
		Abdessamad Belfakir\thanks{\textcolor{blue}{abdobelfakir01@gmail.com}}\\School of Physics and Astronomy, University of Nottingham, Nottingham, NG7 2RD, UK}

	\date{\vspace{-5ex}}
	\maketitle%%%%%%%%%%%%%%%%%%%%%%%%%%%%%%%%%%%%%%%%%%%%%%%
	\begin{abstract}
In this paper, we propose a novel probabilistic control framework for efficiently controlling an ensemble of quantum systems that can also compensate for the interaction of the systems with the external environment. The main challenge in this problem is to simultaneously steer an ensemble of systems with variation in their internal parameters from an initial state to a desired final state. The minimisation of the discrepancy between the probabilistic description of the dynamics of a quantum ensemble and a predefined desired probabilistic description is the key step in the proposed framework. With this objective, the derived solution will not only allow the transitioning of the ensemble from one state to another, but will generally allow steering an initial distribution of the ensemble to a final distribution. Numerical results are presented, demonstrating the effectiveness of the proposed probabilistic control framework.

	\end{abstract}
	{\quotation\noindent{\bf Keywords:}  Probabilistic control, 
		Kullback-Leibler Divergence (KLD), quantum control.
		
		\endquotation}
	
\section{Introduction }\label{sec1}
Designing and implementing control solutions to simultaneously steer a large quantum ensemble between states of interest is a critical aspect of many advanced quantum technologies. These technologies include laser cooling, nuclear magnetic resonance spectroscopy, magnetic resonance imaging, quantum computation \cite{Cory}, and long-distance quantum communication \cite{Duan}. Quantum ensembles consist of a vast number of quantum systems (e.g., spin systems \cite{Bensky,Li,Glaser,Altafini2}) with variations in their parameters. The primary challenge in controlling ensembles is to derive a common external control field capable of concurrently evolving each system in the ensemble from an initial state to a desired final state, by adjusting global parameters of the overall system, rather than controlling individual members.

Control strategies for quantum ensembles have been successfully applied to molecular ensembles \cite{Turinici} and have proven effective in quantum information processing \cite{Khaneja}. Recent studies have focused on controlling inhomogeneous spin ensembles and stabilizing spin system ensembles using Lyapunov control methodology \cite{Beauchard}. Ensemble controllability has been discussed in \cite{Li09}, emphasizing the importance of Lie brackets and noncommutativity in designing compensating controls. Further research has investigated the development of unitary control in homogeneous quantum ensembles for maximizing signal intensity in coherent spectroscopy \cite{Glaser} and employing the sampling-based learning control (SLC) approach \cite{spin,Chen17}.

Environmental decoherence, a crucial aspect of quantum mechanics, arises from the unavoidable interaction between a system and its environment. While extensively studied in controlling and manipulating quantum systems \cite{Jirari,Cui,Hohenester,Gordon,Grace,Grigorenko}, the issue of environmental decoherence has not yet been adequately addressed within the context of quantum ensembles. Decoherence poses a significant challenge for quantum control development, as it demands precise control of system dynamics and results from the interaction between a quantum system and its noisy environment \cite{Zurek,Joos,Schlosshauer}. Protecting quantum system interferences forms the foundation of quantum information applications \cite{Nielsen}. Consequently, the ability to control and manage environmentally induced decoherence would provide substantial benefits for controlling numerous physical and chemical phenomena. Prior work on quantum optimal control in the presence of decoherence has utilized adaptive feedback control methods \cite{Brif_2001} and control optimization algorithms to maximize population transfer fidelity or quantum gates. Genetic algorithms have been proposed to control and suppress decoherence effects \cite{Zhu_2003}, and experimental control of decoherence in molecular vibrational wave packets has been conducted in the laboratory \cite{Branderhorst_2008}.

Given the current state of the art in controlling quantum ensembles, there is an urgent need to develop more accurate control methods that enable the simultaneous steering of ensemble members to the same desired state in the presence of uncertainty or field noise introduced by the system's Hamiltonian or its interaction with the environment. One highly efficient control method for uncertain stochastic classical control systems is the fully probabilistic control approach \cite{Karny_1,RH_2011,RH_2013,Karny_2}. This approach characterizes the generative models of classical system dynamics using probabilistic descriptions. A control law is then derived to minimize the discrepancy between the probabilistic description of the joint distribution of the system dynamics and its controller and a predefined desired joint distribution \cite{RH_2011,RH_2013}. This method has been demonstrated on various classical systems and has proven effective in deriving accurate control laws under high levels of uncertainty and stochasticity \cite{RH_2015,RH_2018,RH_2020}. Consequently, extending and adapting the fully probabilistic control method to control quantum systems is an appealing prospect.

In this paper, we focus on controlling ensembles of finite-dimensional quantum systems that interact with external environments using the fully probabilistic approach. Our goal is to develop a probabilistic control framework, building upon its classical counterpart, capable of accurately steering the members of a quantum ensemble to a desired state value while considering the interaction of the quantum ensemble with the external environment in the optimal control law design.

Extending the fully probabilistic control approach to the quantum mechanical regime presents numerous challenges. A significant challenge is the coupling of a quantum system with its external environment \cite{Kraus,Breuer,Gardiner}. To address this issue, we employ the Lindblad master equation formalism. This formalism models the dynamics of open quantum systems, providing a framework that incorporates the effects of environmental decoherence and dissipation into the system's evolution \cite{Kraus,Neumann}. Another hurdle is the complex nature of quantum system state variables and operators. We tackle this by characterizing generative probabilistic models based on complex probability density functions, crucial for accurately capturing the unique features of quantum systems and ensuring the effectiveness of our probabilistic control framework. A key measure of success in our proposed control methodology is the ability to guide each member of a quantum ensemble to a predetermined target state. Achieving this necessitates the measurement of the final state of the quantum systems in the ensemble and its comparison with the desired state. However, measuring a quantum state is not as straightforward as in classical systems. Quantum systems, defined by their wave functions or density operators, are generally not directly observable. Instead, one must perform a series of measurements on many identically prepared quantum systems to estimate their state, a process known as quantum state tomography. This process aids in estimating the system's density matrix. Intricately tied to this measurement process is the phenomenon of measurement backaction, another significant challenge. This unique quantum trait, where the act of measurement can alter the state of the system, could influence the post-control state of the system, potentially affecting the control methodology's performance. While intriguing, we opt not to consider measurement backaction in this paper to maintain the simplicity of our presentation of the proposed fully probabilistic approach. However, it's worth noting that a similar approach to environmental interaction could be employed to account for measurement backaction \cite{Brun,Doherty}.

Consequently, our primary focus is on the mathematical formulation of fully probabilistic control within quantum settings, as well as the analysis of a class of bilinear quantum control systems. We aim to establish a rigorous probabilistic control framework for quantum ensembles that takes into account environmental decoherence, uncertainty, and the complex nature of quantum state variables. This framework will facilitate the development of more robust and accurate control methods for a variety of quantum technologies. By presenting our approach in a simplified manner and systematically addressing the key challenges, we provide a solid foundation for future work that could extend our proposed approach to include the effects of measurement backaction and other complex phenomena.

%We will study the efficiency of the proposed probabilistic controllers in the management of the environmentally induced decoherences. To be specific we are interested in studying the control of ensembles of spin systems and $\Lambda$-type atomic systems in interaction with external environments described by Linblad operators \cite{spin}. 
%
%This paper describes a new  efficient systematic methodology for  design of controls of quantum ensembles creating specific state to the state transition. We particularly examine ensemble of systems in interaction with an external environment and an external electric field. This methodology consists of two steps (i) Training and (ii) Testing. In the training step we assume that the system  can be described by  a pdf that takes into consideration all sources of noise and uncertainties. Similarly  the  properties of the quantum physical systems can only be described by a pdf. Hence the control problem can  be solved by minimizing the distance between the actual pdf at a given time $t$ and the one associated to the desired state.  These deduced
%controls will  be used in the testing step to control randomly selected number of  members to demonstrating the  performance of the controls and the efficiency of our methodology.  

This paper is organised as follows: In Section (\ref{Evolution of open quantum systems}) we briefly recall the evolution of open quantum systems and develop  the corresponding state space model using the vectorization of their corresponding density operator. In Section (\ref{Fully Probabilistic Control for Quantum Systems}), we introduce a general theory to fully control  quantum systems affected by state dependent noisy environment  using a probabilistic approach and demonstrate its general solution. Then, we apply the developed approach   in Section (\ref{QuantumCA}) to systems described by Gaussian pdfs. In Section (\ref{Results and discussions}) we apply the method to control an ensemble of particular systems and show the effectiveness and the applicability of the method. Finally, we conclude  in Section (\ref{conclu}).

\section{Problem formulation and Quantum systems models}\label{Evolution of open quantum systems}
There are numerous dynamical models describing quantum systems interacting with external environments, depending on the type of system-environment coupling \cite{Kraus,Breuer,Gardiner}. %When there is no initial coupling between the system and environment, the evolution of the system's reduced density operator $\rho \in \mathbf{C}^{l \times l}$ is described by a completely positive, trace-preserving map:
%\begin{equation}
%\rho(t)=\phi(\rho_0),
%\end{equation}
%where $\rho_0$ is the initial density operator, and the map $\phi$, called the Kraus map, is provided in terms of $N$ Kraus operators $K_j \in\mathbf{C}^{l \times l}$ by:
%\begin{equation}
%\phi(\rho_0)=\sum_{j=1}^{N} K_{j}(t) \rho_{0} K_{j}^{\dagger}(t),
%\end{equation}
%where \textcolor{blue}{ $ K_{j}^{\dagger}$ is the adjoint of $K_{j}$}, $\sum_{j=1}^{N} K_{j}(t) K_{j}^{\dagger}(t) = \mathbb{1}_l$, and $\mathbb{1}_l$ denotes the identity operator on the Hilbert space of dimension $l$ \cite{Kraus}.
An alternative approach to describe the time evolution of the open quantum system is through master equations, the form of which depends on the character of the system-environment interaction. Under the Markovian assumption, an open quantum system can be described by the following master equation:
\begin{equation}\label{LVN_eq}
i \dfrac{d\rho(t)}{d t} =[H,\rho(t)]+\mathcal{L}(\rho(t)), \hspace*{0,2cm} \rho(0)=\rho_0,
\end{equation}
where $H$ is the Hamiltonian of the system and $\mathcal{L}[\rho(t)]$ is the open system super-operator. In this paper, we set $\hbar=1$ and consider Hamiltonians of the form:
\begin{equation}
\label{Hamil}
H=H_0+ u(t)H_1,
\end{equation}
where $H_0$ is the Hamiltonian of the isolated free system, and the term $u(t)H_1$ is associated with the time-dependent control Hamiltonian defined through the interaction of the system with the external field $u(t)$. In the Lindblad approach, the open system Liouvillian is:
\begin{equation}\label{Linbad}
\mathcal{L}(\rho(t))=i\sum_{s}\big(L_s\rho(t) L_s^\dagger-\dfrac{1}{2}\{L_s^\dagger L_s,\rho(t)\}\big),
\end{equation}
where $L_s$ are Lindblad operators and $L_s^\dagger$ is the adjoint operator of $L_s$. The time evolution of a physical property described by a Hermitian operator $\hat{o}$ is given by:
\begin{equation}\label{obervable_evolution}
{o}(t)=\braket{\hat{o}}=\Tr(\rho(t)\hat{o}).
\end{equation}

Given the definition of the master equation (\ref{LVN_eq}), we will now demonstrate that the open quantum system defined by this master equation can be equivalently described using a time-dependent equation that involves a state vector representing the vectorization of the density operator. For this purpose, the free Hamiltonian $H_0$ can be expanded element-wise as follows:
\begin{equation}
H_0=\sum_{k=0}^{l-1}E_k\ket{k}\bra{k},
\end{equation}
where $E_k$ is the $k$th eigenvalue of $H_0$ associated with the eigenvector $\ket{k}$, with $k \in {0,1,\dots,l-1}$. Similarly, the Lindblad operators appearing in (\ref{Linbad}) can be written element-wise as follows:
\begin{equation}\label{Gammas}
L_s=L_{j,k}=\sqrt{\Theta_{k\to j}}\ket{j}\bra{k},
\end{equation}
where $\Theta_{k\to j}$ is the dissipative transition rate from the eigenstate $\ket{k}$ to the eigenstate $\ket{j}$. Following this element-wise presentation, the master equation (\ref{LVN_eq}) can be equivalently written as \cite{Breuer}:
\begin{equation}\label{rho_elements}
\dfrac{d\rho_{n,q}(t)}{d t}=(-i{(E_n-E_q)}-\gamma_{n,q})\rho_{n,q}(t)+\sum_{k=0}^{l-1}\Theta_{k\to n}\rho_{k,k}(t)\delta_{n,q}+i{u(t)}\sum_{k=0}^{l-1}\bigg(\rho_{n,k}(t)\bra{k}H_1\ket{q}-\bra{n}H_1\ket{k}\rho_{k,q}(t)\bigg),
\end{equation}
where,
\begin{equation}
\gamma_{n,q}:=\dfrac{1}{2}\sum_{j=0}^{l-1}(\Theta_{n\to j}+\Theta_{q\to j}),
\end{equation}
for $n,q={0,1\dots,l-1}$. Using (\ref{rho_elements}), the vectorization $\tilde{x}$ of the density operator $\rho(t)$ provided in Appendix (\ref{vect_app}) is shown to satisfy the following differential equation:
\begin{equation}\label{NLVN_pa}
\dfrac{d\tilde{x}(t)}{d t} =(\tilde{A}+iu(t) \tilde{N})\tilde{x}(t),\hspace{0.5cm} \tilde{x}(0)=\tilde{x}_0,
\end{equation}
where the matrix elements of the operators $\tilde{A} \in\mathbf{C}^{l^2 \times l^2}$ and $\tilde{N} \in\mathbf{C}^{l^2\times l^2}$ can be easily found from (\ref{rho_elements}). The vectorization of the initial density operator $\rho_0$ is given by $\tilde{x}_0$ \cite{wavepack1,redu}. Considering the following transformation:
\begin{equation}
x(t)\to\tilde{x}(t)-x_e,
\end{equation}
the time evolution equation (\ref{NLVN_pa}) can be rewritten as follows:
\begin{equation}\label{state_space}
\dfrac{dx({t})}{d t}=\tilde{A} x({t})+\tilde{B}(x(t)) u(t) ,\hspace{0.5cm} x(0)=\tilde{x}_0-x_e,
\end{equation}
where $x_e$ is an eigenvector of $\tilde{A}$ defined by $\tilde{A}x_e=0$ and $\tilde{B}(x(t)) =i\tilde{N}(x(t)+x_e)$. This equation, known as the input equation, describes the dependence of the system's dynamics, represented by $x(t)$, on the input electric field ${u(t)}$ \cite{redu}. Thus, by solving the ordinary differential equations (\ref{state_space}), we can accurately determine the evolution of the open quantum system. The solution to (\ref{state_space}) can be expressed as:
\begin{align}\label{StateQua0}
x(t+1) = A x(t) + B(x(t)) u(t),
\end{align}
where the matrices $A$ and $B(x({t}))$ are defined as:
\begin{equation}\label{A_B_t}
A = e^{\tilde{A} \Delta t}, \hspace{0.5cm}\text{and}\hspace{0.5cm}
B(x({t}))= \bigg( \int_0^{\Delta t} e^{\tilde{A} \lambda} \tilde{B}(x(\lambda)) \mathrm{d} \lambda\bigg),
\end{equation}
with $\lambda = \Delta t - t$ and $\Delta t$ representing the sampling period. Equivalently, relation (\ref{StateQua0}) can be rewritten as:
\begin{align}\label{equation_15}
x({t} )= A x({t-1}) + B(x({t-1})) u(t-1).
\end{align}
In practice, control of quantum physical systems can be achieved through the use of multiple electric fields in the control Hamiltonian within the von Neumann equation \cite{spin}. By denoting $x_t\equiv{x}(t)$ and $u(t)\equiv u_t$, the relation (\ref{equation_15}) can be reformulated as:
\begin{align}\label{StateQua}
x_{t} = A x_{t-1} + B(x_{t-1}) u_{t-1}.
\end{align}
In this work, we aim to control an ensemble of inhomogeneous quantum systems, each characterized by a unique Hamiltonian due to the dispersion in the parameters that describe them. This inhomogeneity or dispersion results in each system responding slightly differently to the same control signal, a challenge that needs to be addressed in control problems. To account for this parameter dispersion and its impact on the state evolution of the systems in the ensemble, we introduce multiplicative noise into our model. Unlike additive noise, multiplicative noise varies with the state of the system, making it a suitable choice for modeling uncertainties or variations that scale with the system's state. Therefore, we add the multiplicative noise term, $\eta(x_{t-1})$, to the discretized equation (\ref{StateQua}). The process or measurement noise levels in $\eta(x_{t-1})$ depend on the system state vector, resulting in a state-dependent noise term. The final form of our model is given by:
\begin{align}\label{Inpu2NoiseDelay}
x_{t} = A x_{t-1} + B(x_{t-1}) u_{t-1} + \eta(x_{t-1}).
\end{align}
where,
\begin{equation}\label{eta} \eta(x_{t-1}) =\zeta_t Ax_{t-1},
\end{equation}
and $\zeta_t$ is a scalar noise. This modified equation (\ref{Inpu2NoiseDelay}) incorporates the effects of multiplicative noise and thus accommodates the parameter dispersion inherent in our ensemble of inhomogeneous quantum systems. By employing a vectorized description of the density matrix, we can now reformulate the time evolution of the observable $\hat{o}$, given in (\ref{obervable_evolution}), in terms of the vectorized density operator. These adaptations to the model, both in terms of multiplicative noise and the vectorized description of the density matrix, allow us to better capture the unique dynamics of each member of the quantum ensemble, enhancing the effectiveness of our control methodology. To further improve the precision of our model, we consider additional sources of noise and uncertainties that could affect the observed values $o_t$. This consideration leads to the inclusion of the noise term $\sigma_t$ in our output equation:
\begin{equation}\label{outp}
o_t=Dx_t+\sigma_t.
\end{equation}
Here, $D=(\text{vec}(\hat{o}^T))^T$ represents the matrix obtained by transposing the vectorized version of the observable operator $\hat{o}$, with "vec" denoting the vectorization operator and $\hat{o}^T$ being the transpose of $\hat{o}$. The term $\sigma_t$ is a multivariate Gaussian noise, introduced to model various sources of noise and uncertainties, such as measurement error, environmental disturbances, and process noise \cite{Doherty}. The integration of this term helps enhance the realism and robustness of our model, thereby improving its predictive accuracy and control performance.

By defining the quantum system's dynamics through both the input (\ref{Inpu2NoiseDelay}) and output (\ref{outp}) equations, we establish a bilinear state-space model that effectively characterizes the behavior of the quantum systems in our ensemble. This comprehensive model provides a strong foundation for the subsequent design of effective control strategies.

A key innovation introduced in this paper lies in the extension of the bilinear state-space model, originally formulated in previous works \cite{Doherty,wavepack1,redu}, to incorporate both multiplicative and additive Gaussian noise. By transforming the quantum open system's evolution, typically described by the von Neumann-Lindblad equation, into this enhanced state-space model, we provide a novel representation that not only offers a more comprehensive picture of the quantum system's dynamics but also explicitly accounts for various sources of uncertainties such as manufacturing differences, different environmental conditions, the inherent uncertainty in quantum mechanics, sampling, parameter, and functional uncertainties. These are captured by the inclusion of noise terms in both the input and output equations. This innovative formulation serves as the foundation of the fully probabilistic control framework for quantum systems proposed in this paper, setting the stage for the introduction of advanced control strategies in the following sections.

In the sections to follow, we will demonstrate how utilizing this bilinear state-space model with noise terms enables the development of efficient control strategies within a fully probabilistic framework. This approach provides a robust characterization of the system's time evolution, specifically tailored for quantum systems, and allows for enhanced control and management of their dynamics. Given the precision and robustness of this methodology, it holds significant potential for a wide range of applications in quantum information processing, quantum communication, and quantum sensing. In these domains, precise control and estimation of quantum states are of paramount importance \cite{Ryan,Khaneja2}..

\section{Fully Probabilistic Control for Quantum Systems}\label{Fully Probabilistic Control for Quantum Systems}
Building upon the stochastic bilinear state-space model formulated in the previous section, we now present the control objective and introduce the general framework for the fully probabilistic control of ensembles of quantum systems. This innovative framework is designed to account for various sources of uncertainty and stochasticity affecting the quantum systems in a constructive manner, ultimately yielding more accurate control laws under these challenging conditions.

\subsection{Objectives for the Fully Probabilistic Control of Quantum Systems}
As previously discussed, this paper focuses on controlling ensembles of quantum systems, where individual elements are subject to dispersion in their parameters. These uncertainties are incorporated into the stochastic bilinear state-space model represented by equations (\ref{Inpu2NoiseDelay}) and (\ref{outp}). Consequently, the state of the system at each time step can be represented by the following conditional probability density function (pdf):
\begin{equation}
\label{s_x_t}
s(x_{t}\left| {x_{t - 1}},u_{t-1} \right.).
\end{equation}
This representation provides a comprehensive description of the current state of the system $x_t$, as influenced by the previous states of the system $x_{t-1}$ and electric field $u_{t-1}$. Similarly, due to the presence of the noise term $\sigma_t$ in equation (\ref{outp}), the state of the measurement $o_t$ can be characterized by the following pdf:
\begin{align}
\label{s_o_t}
s({o_t}\left| {{x_{t}}} \right.).
\end{align}
Armed with these conditional pdfs, the objective of the fully probabilistic control can be stated as designing a controller's pdf, $c(u_t|x_t)$, that minimizes the Kullback-Leibler divergence (KLD) between the joint pdf of the closed-loop system dynamics, $f(\mathcal{Z}(t, \mathcal{H}))$, and a pre-specified ideal joint pdf, $^If(\mathcal{Z}(t, \mathcal{H}))$:
\begin{equation}\label{KLD}
\mathcal{D}(f||^I f)=\int f(\mathcal{Z}(t, \mathcal{H}))\ln\big(\dfrac{f(\mathcal{Z}(t, \mathcal{H}))}{^If(\mathcal{Z}(t, \mathcal{H}))}\big)d\mathcal{Z}(t, \mathcal{H}),
\end{equation}
where
\begin{equation}\label{JointDist}
f(\mathcal{Z}(t, \mathcal{H})) = \prod_{t=1}^\mathcal{H} s(x_t|x_{t-1},u_{t-1})s(o_t|x_{t})c(u_{t-1}|x_{t-1}),
\end{equation}
and
\begin{equation}\label{IdealJointDist}
^If(\mathcal{Z}(t, \mathcal{H}))=\prod_{t=1}^\mathcal{H} {^Is(x_t|x_{t-1},u_{t-1})} {^Is(o_t|x_{t})}{^Ic(u_{t-1}|x_{t-1})}.
\end{equation}
In these equations, $^Is(x_t|x_{t-1},u_{t-1})$ denotes the ideal pdf of the state vector $x_t$, $^Is(o_t|x_{t})$ denotes the ideal pdf of the measurement $o_t$, $\mathcal{Z}(t, \mathcal{H})=\{x_t, \dots, x_\mathcal{H}, o_t, \dots, o_\mathcal{H},u_{t-1}, \dots, u_\mathcal{H}\}$ is the closed-loop observed data sequence and $ \mathcal{H} < \infty$ is a given control horizon, and $^Ic(u_{t-1}|x_{t-1})$ represents the ideal pdf of the controller. Given that the state $x_t$ is not directly observable and that the output pdf is conditioned on the state $x_t$ of the quantum system, it is reasonable to allow the state of the quantum system to evolve naturally. Consequently, we assume that $^Is(x_t|x_{t-1},u_{t-1}) = s(x_t|x_{t-1},u_{t-1})$. This assumption simplifies the problem and focuses the controller design on accounting for uncertainties in the measurements and control inputs, rather than attempting to manipulate the natural evolution of the quantum system. By doing so, the proposed fully probabilistic control framework remains compatible with the inherent characteristics of quantum systems while effectively addressing the challenges associated with uncertainties and parameter dispersion in ensembles of quantum systems.

Previous works, such as those cited in \cite{Karny_1,RH_2011,RH_2013}, have demonstrated that the minimization of the KLD equation (\ref{KLD}) with respect to the control input is achieved by initially defining $-\ln(\gamma(x_{t-1}))$ as the expected cost-to-go function. In the context of this work, we express $-\ln(\gamma(x_{t-1}))$ as follows:
\begin{align}\label{OptPer}
-\ln(\gamma(x_{t-1})) &= \underset{c(u_{t-1}|x_{t-1})}{\text{min}}\sum_{\tau=t}^\mathcal{H} \int f(\mathcal{Z}_t, \dots, \mathcal{Z}_\mathcal{H}|x_{t-1}) \times\ln \left( \dfrac{s(o_\tau|x_\tau)c(u_{\tau-1}|x_{\tau-1})}{^Is(o_\tau|x_\tau)^Ic(u_{\tau-1}|x_{\tau-1})} \right) d(\mathcal{Z}_t, \dots, \mathcal{Z}_\mathcal{H}),
\end{align}
for any arbitrary $\tau \in {1, . . . , \mathcal{H}}$. In this equation, $\mathcal{Z}_t$ denotes the set ${x_t,o_t,u_{t-1}}$. Importantly, observe that the argument of the $\ln$ function does not include the conditional density of the quantum system state, $s(x_t|x_{t-1},u_{t-1})$, nor its corresponding ideal distribution, $^Is(x_t|x_{t-1},u_{t-1})$. This omission is not an oversight, but instead reflects a fundamental quantum property: the state of the dynamic system, being unobservable, evolves naturally according to the laws of quantum evolution. This quantum characteristic is manifested by equating the actual distribution of the quantum system state to its corresponding ideal distribution, i.e., $^Is(x_t|x_{t-1},u_{t-1}) = s(x_t|x_{t-1},u_{t-1})$. This equality underpins the unobserved natural evolution of the quantum state, a principle fundamental to quantum dynamics and a cornerstone of the proposed control methodology.

 Minimization of the cost to go function (\ref{OptPer}) can then be performed recursively to give the following recurrence equation analogical to the dynamic programming solution:
\begin{align}
	\label{cost_to_go_2}
-\ln(\gamma(x_{t-1})) &=  \underset{c(u_{t-1}|x_{t-1})}{\text{min}} \int\bigg[ s(x_t|x_{t-1}, u_{t-1})s(o_t|x_t)c(u_{t-1}|x_{t-1})\bigg(\ln \left( \dfrac{s(o_t|x_t)c(u_{t-1}|x_{t-1})}{^Is(o_t|x_t)^Ic(u_{t-1}|x_{t-1})} \right) \nonumber \\ &-\ln(\gamma(x_{t}))\bigg)  \bigg]d(x_t,o_t,u_t).
\end{align}
This innovative approach for devising a controller pdf provides an elegant solution for addressing the inherent uncertainties in quantum systems. By minimizing the KLD between the actual and ideal joint pdfs, our method facilitates more accurate control laws under various conditions of uncertainty, thereby enhancing the overall performance and reliability of quantum systems. The fully probabilistic control framework enables us to account for various sources of uncertainty and stochasticity affecting the quantum systems, and to construct efficient control strategies tailored specifically for quantum systems. Moreover, the recursive nature of the controller pdf design allows for efficient online implementation, making it adaptable to real-time applications in quantum information processing, quantum communication, and quantum sensing. The proposed approach addresses the challenges associated with controlling ensembles of quantum systems and offers a novel, comprehensive framework for the robust management of their dynamics.

\subsection{General Solution to the Fully Probabilistic Quantum Control Problem}\label{SFPDCQS}
In this section, we derive the general solution of the optimal randomized controller that minimizes the KLD given in equation (\ref{KLD}) for systems defined by arbitrary pdfs. This solution is not restricted to specific forms of the pdfs and can be applied to any quantum system described by arbitrary conditional pdfs.
\begin{theorem}\label{Prop1}
The pdf of the optimal control law $c(u_{t-1}|x_{t-1})$ that minimizes the cost-to-go function (\ref{cost_to_go_2}) between the joint pdf of the closed loop description of a quantum system (\ref{JointDist}) and a pre-specified ideal one (\ref{IdealJointDist}) is given by:
\begin{equation}
	\label{eq:Eqn4}
	c(u_{t-1}|x_{t-1})=\frac{^Ic(u_{t-1}|x_{t-1}) \exp [-\beta(u_{t-1},x_{t-1})]}{\gamma(x_{t-1})},
\end{equation}
where,
\begin{equation}\label{gamma2}
\gamma(x_{t-1}) = \int {^Ic(u_{t-1}|x_{t-1}) \exp[{-\beta(u_{t-1},x_{t-1})}] \mathrm{d} u_{t-1}},
\end{equation}
and,
\begin{equation}\label{beta}
\beta(u_{t-1},x_{t-1}) = \int s(x_{t}|u_{t-1}, x_{t-1}) s(o_{t}|x_{t})\times \ln \bigg ( \frac{s(o_{t}|x_{t}) }{ ^I s(o_{t}|x_{t}) } \frac{1}{{\gamma}(x_{t})}\bigg) \mathrm{d} x_{t}\mathrm{d} o_{t}.
\end{equation}
\end{theorem}
\begin{proof}
This theorem can be easily proven by adapting the proof of Proposition 2 in Ref \cite{Karny_2}.
\end{proof}
The above theorem presents a general solution for the fully probabilistic control problem of quantum systems, without being restricted to a specific form of the generative probabilistic model describing the system dynamics. This solution is applicable to any quantum system described by an arbitrary pdf.

\section{Solution to Quantum Control Problems with Gaussian pdfs}\label{QuantumCA}
Building upon the general solution of the fully probabilistic control approach presented in Section (\ref{SFPDCQS}), we now focus on quantum systems characterized by Gaussian probability density functions (pdfs). The linear Gaussian setting simplifies the problem and provides an opportunity to demonstrate the efficacy of the method.

\subsection{System Description and Assumptions}
This section considers quantum systems described by a bilinear state space model of the form given in (\ref{Inpu2NoiseDelay}) and (\ref{outp}) when the noises $\eta_t$ in (\ref{eta}) and $\sigma_t$ are generated using a Gaussian process. Under these conditions, the pdf of the quantum system state (\ref{s_x_t}) can be described by linear Gaussian pdfs as follows:
\begin{align}
\label{eq:eq4}
s(x_{t}\left| {x_{t - 1}},u_{t-1} \right.) &\sim \mathcal{N}_{\mathcal{C}}(\mu_{t},{\Gamma_t}),
\end{align}
where $\mu_t$ and $\Gamma_t$ are the mean and covariance matrices, respectively:
\begin{align}
\label{eq:relation1}
\mu_t& =E(x_t)=Ax_{t-1}+Bu_{t-1},\\
\Gamma_t&= E((x_t-\mu_t)(x_t-\mu_t)^\dagger), \nonumber \\
& = E(\zeta_t A x_{t-1} \zeta_t x_{t-1}^\dagger A^\dagger) ,\nonumber \\
\label{eq:InDepCov}
& = A x_{t-1} \Sigma x_{t-1}^\dagger A^\dagger,
\end{align}
where we used equations (\ref{Inpu2NoiseDelay}) and (\ref{eta}) in order to evaluate the covariance matrix $\Gamma_t$. In this case, the matrices $A$ and $B$ represent the state and control matrices, respectively, appearing in (\ref{Inpu2NoiseDelay}). The functional $E(.)$ denotes the expectation value, $x_t^\dagger$ is the conjugate transpose of $x_t$, $\Sigma=E(\zeta_t^2)$ and $\mathcal{N}_{\mathcal{C}}$ denotes a complex Gaussian distribution. The forms of the complex Gaussian distribution and the normal Gaussian distribution are recalled in Appendix (\ref{app_C}).

Furthermore, with the noise term $\sigma_t$ in (\ref{outp}) generated by a Gaussian distribution, the pdf of the measurement $o_t$ in (\ref{s_o_t}) can be characterized by Gaussian pdf as follows:
\begin{align}
\label{eq:eq5}
s({o_t}\left| {{x_{t}}} \right.) &\sim \mathcal{N}(o_d,G),
\end{align}
where $o_d=Dx_t$ and $G=E(\sigma_t \sigma_t^T)=E((o_t-Dx_t)(o_t-Dx_t)^T)$ are the mean and covariance matrices, respectively of this pdf. 

Under the aforementioned conditions, the quantum system state $x_t$, measurement $o_t$, and the electric field $u_t$ at each time $t$ can be fully described by the following joint pdf:
\begin{equation}\label{gpdf2}
f(x_t,o_t,u_{t-1}|x_{t-1})=c(u_{t-1}|x_{t-1})\mathcal{N}(Dx_{t},G)\mathcal{N}_{\mathcal{C}}(\mu{t},{\Gamma_t}).
\end{equation}
Taking this into account, we postulate that the ideal joint pdf governing the entire system-- including the system state, measurement, and controller-- is given by:
\begin{equation}\label{igpdf2}
^If(x_t,o_t,u_{t-1}|x_{t-1})= {}^Is(x_{t} | x_{t-1},{u_{t-1}}){}^Is({o_t}\left| {{x_{t}}} \right.) {}^Ic({u_{t-1}}\left| {x_{t - 1}} \right),
\end{equation}
where the ideal distributions are given by:
\begin{align}
\label{eq:eq1st}
		{}^Is(x_{t} | x_{t-1},{u_{t-1}})& =s(x_{t} | x_{t-1},{u_{t-1}})\sim  \mathcal{N}_{\mathcal{C}}(\mu_{t},{\Gamma_t}),\\
	\label{eq:eq2}
	{}^Is({o_t}\left| {{x_{t}}} \right.) &\sim \mathcal{N}({o}_{d},G_r),\\
	\label{eq:eq3}
	{}^Ic({u_{t-1}}\left| {x_{t - 1}} \right.) &\sim \mathcal{N}(u_r,\Omega),
\end{align}
and where $o_d$ and $G_r$ are the mean and covariance matrices of the ideal pdf of the measurement respectively, and $u_r$ and $\Omega$ are the mean and the covariance matrices respectively of the ideal pdf of the controller. Furthermore, as illustrated in Equation (\ref{eq:eq1st}), the ideal distribution of the quantum state, ${}^Is(x_{t} | x_{t-1},{u_{t-1}})$, is set to be identical to the actual distribution, $s(x_{t} | x_{t-1},{u_{t-1}})$, which is characterized by the master equation (\ref{LVN_eq}). This equivalence signifies the undetected natural evolution of the quantum state, an intrinsic characteristic of quantum dynamics, which forms a pivotal component of the proposed control methodology.

\subsection{Solution to the linear Gaussian Quantum Control Problem}
Given the defined joint pdf of the quantum system (\ref{gpdf2}) and the ideal joint pdf (\ref{igpdf2}), we can now apply the general solution of the fully probabilistic control approach presented in Section (\ref{SFPDCQS}) to solve the quantum control problem for linear Gaussian systems. Specifically, by incorporating the components of the joint pdf of the closed-loop description of the quantum system (\ref{gpdf2}) and the ideal joint pdf (\ref{igpdf2}) into equations (\ref{eq:Eqn4})-(\ref{beta}), we can derive analytic solutions for the optimal cost-to-go function and the randomized controller. We will soon present the form of the pdf for the optimal control, but first, let us outline the structure of the optimal cost-to-go function.

\begin{theorem}\label{Theo1}
By substituting the ideal distribution of the measurements (\ref{eq:eq2}), the ideal distribution of the controller (\ref{eq:eq3}), the real distribution of the measurement (\ref{eq:eq5}), and the real distribution of system dynamics (\ref{eq:eq4}) into (\ref{gamma2}), the optimal cost to go function can be shown to be given by,
	\begin{align}
		- \ln \left( {\gamma \left( {{x_{t-1}}} \right)} \right) &= 0.5{x}_{t-1}^{\dagger}{M_{t-1}}{x_{t-1}}+0.5P_{t-1}{{x}_{t-1}} + {0.5\omega_{t-1}}, \label{49}
	\end{align}
	where, 
	\begin{align}
		\label{Mt}
		{M_{t - 1}} &=A^{\dagger} (D^{\dagger} G_r^{-1}D+M_{t})A-A^{\dagger} (D^{\dagger} G_r^{-1}D+M_{t})^{\dagger} B
		\big(\Omega^{-1}+B^{\dagger} ({D^{\dagger} }G_r^{ - 1}D+M_{t})B\big)^{-1}
		\nonumber\\
		&	B^{\dagger} (D^{\dagger} G_r^{-1}D+M_{t})A+A^{\dagger} (D^{\dagger} G_r^{-1}D+M_{t}){\Sigma} A,\end{align}

	\begin{align}
		\label{Pt}
		P_{t-1}&=(P_{t}-2{o}_{d}^{\dagger} G_r^{-1}D)A+2(\Omega^{-1}u_r-0.5B^{\dagger} (P_{t}^{\dagger} -2D^{\dagger} G_r^{-1}{o}_{d}))^{\dagger} 
		\big(\Omega^{-1}+B^{\dagger} ({D^{\dagger} }G_r^{ - 1}D+M_{t})B\big)^{-1}\nonumber\\
		&B^{\dagger} (D^{\dagger} G_r^{-1}D+M_{t})A,
	\end{align}
	and,
	\begin{align}
		\label{Const}
		\omega_{t-1}&={\omega_{t}}+ {o}_{d}^{\dagger} G_r^{ - 1}{{o}_{d}}   +\ln\bigg(\dfrac{|G_r|}{|G|}\bigg)- \Tr\big(G(G^{-1}-G_r^{ - 1})\big)
		+u_r^{\dagger} \Omega^{-1}u_r	\nonumber\\
		&-\big( \Omega^{-1}u_r-0.5B^{\dagger} (P_{t}^{\dagger} -2{D}^{\dagger}G_r^{ - 1}{{o}_{d}})\big)^{\dagger} 
		\big(\Omega^{-1}+B^{\dagger} ({D^{\dagger} }G_r^{ - 1}D+M_{t})B\big)^{-1}	
		\big( \Omega^{-1}u_r-0.5B^{\dagger} (P_{t}^{\dagger} -2{D}^{\dagger}G_r^{ - 1}{{o}_{d}})\big)\nonumber\\
		&-2\ln({|\Omega|^{-1/2}}|\Omega^{-1}+B^{\dagger} ({D^{\dagger} }G_r^{ - 1}D+M_{t})B|^{-1/2}),
	\end{align}	
where $|G|$ stands for the determinant of the matrix $G$. 
\end{theorem}
\begin{proof}
	The proof of this theorem is given in Appendix (\ref{AppA}).
\end{proof}
Having established the form of the optimal cost-to-go function (\ref{49}), we can utilize it in equation (\ref{eq:Eqn4}) to determine the distribution of the optimal controller $c(u_{t-1}|x_{t-1})$. This optimizes the pdf of the entire system (\ref{gpdf2}), essentially bringing the pdf of the system as close as possible to the ideal one. The form of this optimal controller is presented in the subsequent theorem.

\begin{theorem}\label{Theo2}
The controller's distribution, which minimizes the Kullback-Leibler Divergence (KLD) between the joint pdf of the closed-loop system dynamics described in (\ref{gpdf2}) and the ideal joint pdf presented in (\ref{igpdf2}), is given by:
	\begin{align}\label{Probcapp}
c({u_{t-1}}\left| {x_{t - 1}} \right.) &\sim \mathcal{N}(v_{t-1},R_t),
	\end{align}
where, 
\begin{align}\label{optimalcapp}
	v_{t-1}&=\bigg({\Omega ^{ - 1}} + {B^{\dagger}}( {D^{\dagger}}G_r^{ - 1}D + {M_t})B\bigg)^{-1} \bigg({\Omega ^{ - 1}}{u_r} - {B^{\dagger}}({D^{\dagger}}G_r^{ - 1}D + {M_t})A{x_{t - 1}} -0.5 {B^{\dagger}}( P_{t }^{\dagger} - 2{D^{\dagger}}G_r^{ - 1}{{o}_{d}} )\bigg),
\end{align}
and,
\begin{align}\label{optimal_variance}
	R_t&=\bigg({\Omega ^{ - 1}} + {B^{\dagger} }( {D^{\dagger} }G_r^{ - 1}D + {M_t})B\bigg)^{-1}.
\end{align}
\end{theorem}
\begin{proof}
The proof is given in Appendix (\ref{appC}).
\end{proof}
The randomised optimal controller derived in Theorem (\ref{Theo2}) represents a novel approach to controlling an ensemble of quantum systems. By introducing a multiplicative noise, $\eta(x_{t-1})$ and an additive noise, $\sigma_{t}$, to the state vector $x_t$ and the observations $o_t$ respectively, this innovative linear form of the controller effectively addresses the challenges of controlling randomly selected members of quantum ensembles. This breakthrough in quantum control theory enhances the efficacy and adaptability of control strategies for  quantum systems.

\subsection{Implementation of the Gaussian controller}
The implementation of the Gaussian controller, as described in Theorem (\ref{Theo2}), can be broken down into a step-by-step process, which is summarized in the following algorithm. The implementation consists of two primary phases: Optimization and Testing.

During the Optimization phase, the objective is to characterize the probability density function (pdf) of the randomized controller by optimizing its parameters. The optimization process utilizes historical data, simulated data, or real-time data obtained online to achieve the desired performance of the controlled quantum system. By optimizing these parameters, the randomized controller will be able to drive the pdf of the joint distribution of the closed-loop description of the quantum system to a pre-specified ideal joint pdf.

The Testing phase involves using the Gaussian controller, constructed during the Optimization phase, to drive members of an ensemble to the same target state. The ensemble members are randomly generated, allowing for an assessment of the controller's performance under a range of initial conditions. This benchmarking process demonstrates the robustness and efficiency of the Gaussian controller when applied to real-world quantum systems.

\begin{algorithm}[H]
\caption{Fully probabilistic control of quantum systems}\label{alg:cap}
\begin{algorithmic}[1]
\State $\textbf{Optimization:}$
\State Evaluate the operator $D$ associated with the target operator $\hat{o}$;
\State Compute the matrices $\tilde{A}$ and $\tilde{N}$ from equation (\ref{rho_elements});
\State Determine the predefined desired value ${o}_{d}$;
\State Specify the initial state $x_0$, the shifting state $x_e$, and then calculate the initial value of the measurement state $o_0 \gets Dx_0$;
\State Provide the covariances $\Sigma, G, G_r$, and $\Omega$ of the multiplicative noise, $\eta(x_{t-1})$, the measurement vector, $o_t$, the ideal distribution of $o_t$, and the controller, $u_t$ respectively;
\State Initialize: $t \gets 0$, $M_0 \gets \text{rand}$, $P_0 \gets \text{rand}$;
\While{$t \ne \mathcal{H}$}
\State Set: $\zeta_t = $ sample\_from\_Gaussian($0$, $\Sigma$);
\State Evaluate $A$ and $B$ using equation (\ref{A_B_t});
\State Calculate the steady state solutions of $M_t$ and $P_t$ following the formulas provided in equations (\ref{Mt}) and (\ref{Pt}), respectively;
\State Use $M_t$ and $P_t$ to compute the optimal control input, $v_{t-1}$ following equation (\ref{optimalcapp}) given in Theorem (\ref{Theo2});
\State Set: $u_{t-1} \gets v_{t-1}$;
\State Use the obtained control input from the previous step to evaluate $x_t$ according to equation (\ref{Inpu2NoiseDelay}), $x_{t} \gets Ax_{t-1} + B(x_{t-1}) u_{t-1} + \zeta_{t}Ax_{t-1}$;
\State Set: $\sigma_t = $ sample\_from\_Gaussian($0$, $G$);
\State Following equation (\ref{outp}), evaluate $o_t$ to find the measurement state at time instant $t$, $o_t \gets Dx_t + \sigma_t$;
\State Update time: $t \gets t + 1$;
\EndWhile
\State $\textbf{Testing:}$
\State Set $N$: the number of members of the ensemble;
\State Initialize: $k \gets 1$;
\While{$k \le N$}
\While{$t \ne \mathcal{H}$}
\State Set: $\zeta^k_t = $ sample\_from\_Gaussian($0$, $\Sigma$);
\State Evaluate $A^k$ and $B^k$ using equation (\ref{A_B_t});
\State Using the obtained control input from the optimization phase, evaluate $x_t^k$ according to equation (\ref{Inpu2NoiseDelay}), $x_{t}^k \gets A^kx^k_{t-1} + B^k(x^k_{t-1}) u_{t-1} + \zeta_{t}^k A^k{x}_{t-1}$;
				\State  Set: $\sigma^k_t = $ sample\_from\_Gaussian($0$, $G$);
				\State  Following equation (\ref{outp}), evaluate $o^k_t$ to find the measurement state at time instant $t$, $o^k_t \gets Dx_t^k + \sigma_t^k$;
			\EndWhile
			\State  Increment ensemble member: $k \gets k + 1$;
		\EndWhile
	\end{algorithmic}
\end{algorithm}
Algorithm (\ref{alg:cap}) will be applied in the next section to control an ensemble of randomly generated systems to the same desired state.

\section{Results and discussions}\label{Results and discussions}
\subsection{Control of quantum ensemble}
In this section, we demonstrate the effectiveness of our algorithm when applied to an ensemble of quantum systems. We consider a simplified Hamiltonian $H$ as described in (\ref{Hamil}), given by:
\begin{equation}
H = H_0 + {u}_{t}H_1,
\end{equation}
where we assume that the system interacts with only one electric field, represented by ${u}_{t}$. Our control objective is to transfer an ensemble of quantum systems, initially prepared in a state $\rho_i$, to a predefined target state $\rho_d$. This transfer is achieved through the interaction between the ensemble members and the optimal randomized electric field. The parameters of this electric field are optimized using Theorem (\ref{Theo2}) to achieve a pre-specified desired joint density function for the joint distribution of the ensemble dynamics. We will demonstrate that the optimized randomized controller is highly effective in controlling all members of the ensemble.

The control problem we aim to address is to maximize the fidelity between the actual state $\rho_t$ and the target state $\rho_d$ for all ensemble members. Here, $\rho_t$ represents the density operator describing the systems at time $t$. By applying our algorithm ({\ref{alg:cap}) to this ensemble of quantum systems, we showcase its efficiency and versatility in controlling complex quantum ensembles. This highlights the potential of our algorithm to tackle real-world quantum control challenges.

In the following, we present the simulation results for the control of the quantum ensemble using our proposed Gaussian controller. The simulations demonstrate the robustness and effectiveness of the controller in driving the ensemble members towards the desired target state. Furthermore, we will discuss the implications of our results and the potential applications of our algorithm in quantum control and quantum information processing tasks.

\subsubsection{spin-1/2 ensemble}
Consider a spin-$1/2$ system interacting with an electric field $u_t$ \cite{spin}. The Hamiltonian describing this interaction is given by:
\begin{equation}\label{58_}
H=H_0+{u}_{t}H_1=\dfrac{1}{2}\sigma_3+\dfrac{{u}_{t}}{2}(\sigma_1+\sigma_2),
\end{equation}
where $\sigma_1$, $\sigma_2$, and $\sigma_3$ are the Pauli operators given in the basis ${\ket{1},\ket{0}}$ by the Pauli matrices,
\begin{equation}\label{PauliMat}
\sigma_1 =\left(
\begin{array}{cc}
0&1\\
1 & 0 
\end{array}
\right),\hspace{0.4cm} \sigma_2=\left(
\begin{array}{cc}
0 &-i\\
i & 0 
\end{array}
\right),\hspace{0.4cm}\sigma_3 =\left(
\begin{array}{cc}
1 &0\\
0 & -1
\end{array}
\right).
\end{equation}
Furthermore, we consider that the system interacts with an environment, e.g., a surrounding vacuum state. The master equation describing this interaction is given by:
\begin{equation}
\dfrac{d\rho(t)}{d t} =-i[H,\rho(t)]+\Theta\bigg(\sigma_{-}\rho(t)\sigma_{+}-\dfrac{1}{2}{\sigma_{+}\sigma_{-},\rho(t)}\bigg), \hspace*{0,2cm} \rho(0)=\rho_0=\ket{0}\bra{0},
\end{equation}
where $\sigma_{\pm}=\dfrac{\sigma_{1}\pm i\sigma_{2}}{2}$ and $\Theta$ describes the coupling between the system and the environment. 

In the optimization phase, the control objective is to transform the system state from its initial value $x_0=(0,1,0,0)^T$ which corresponds to the first eigenstate of the free Hamiltonian, $\ket{0}=(0,1)^T$, to the desired state $x_d=(1,0,0,0)^T$ which corresponds to the second eigenstate of the free Hamiltonian $\ket{1} = (1,0)^T$, using the proposed probabilistic quantum control approach. Consequently, the target operator is taken to be $D=[1\hspace{0.2cm}0 \hspace{0.2cm}0\hspace{0.2cm}0]$, which is equivalent to $D=(\text{vec}(\Pi_1))^T$, where $\Pi_1=\ket{1}\bra{1}$. The matrix elements of the $\tilde{A}$ and $\tilde{N}$ matrices are provided in Appendix (\ref{spinn_1_2}). In the optimization phase of Algorithm (\ref{alg:cap}), we consider $x_e=(0,0,0,0)^T$, $\Delta{t}=2.5\times10^{-6}$, $G_r=0.00001$, $\Omega=10$, and $\Theta=0.1$, and the system is evolved until a sufficient fidelity between the actual state and the target state is achieved. 

Following the optimization phase, the obtained control signal is used to control one of the members of the spin-1/2 system.  Figure (\ref{fig:sub-first}) illustrates the obtained time evolution of the populations of the two states $\ket{0} \bra{0}$ and $\ket{1}\bra{1}$ of the selected member, and Figure (\ref{fig:sub-second}) displays the time evolution of the optimal control signal responsible for achieving the control objective. We can clearly see that the population of the state $\ket{1}\bra{1}$ has reached its desired value, i.e., $o_d=1$, and it is maintained at that value for extended periods. This indicates that the optimal electric field is efficient in preventing the environmentally induced relaxation of the system to the ground state $\ket{0} \bra{0}$, signifying that the optimal control can counteract the natural effect of the environment, i.e., relaxation.

Furthermore, the optimized randomized controller obtained from the optimization phase is applied to a sample of ensemble members, each corresponding to different realizations of the state time evolution generated through the noise term $\eta(x_{t-1})$. This experiment allows us to assess the performance of the controller on a diverse set of members, each representing a unique noise-driven evolution. Figure (\ref{performance}) displays the fidelities between the target state $\ket{1}\bra{1}$ and the states of $1000$ ensemble members interacting with the optimal control electric field. From this figure, we can clearly see that the fidelities for the state transition of all tested 1000 members lie in the interval of $[0.9932,1]$ with a mean value equal to $0.9945$. This demonstrates that the algorithm exhibits exceptional performance in this specific example. It is worth mentioning that we have considered members which are interacting with their external environment. This implies that the optimal electric field is capable of driving the ensemble members to the desired state and preventing their relaxation, even in the presence of environmental effects. The results underscore the robustness and effectiveness of our algorithm in addressing real-world quantum control challenges that involve noise and environmental interactions.  
\begin{figure}[t!]
	\begin{subfigure}{.5\textwidth}
		\centering
		% include first image
		\includegraphics[width=.8\linewidth]{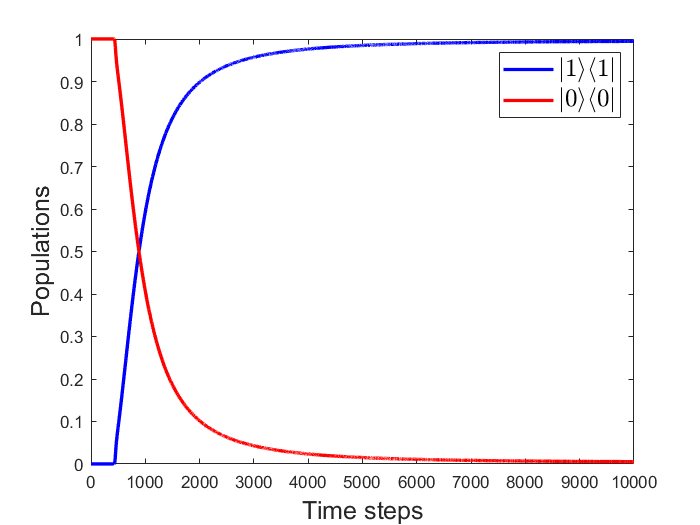}  
		\caption{}
		\label{fig:sub-first}
	\end{subfigure}
	\begin{subfigure}{.5\textwidth}
		\centering
		% include second image
		\includegraphics[width=.8\linewidth]{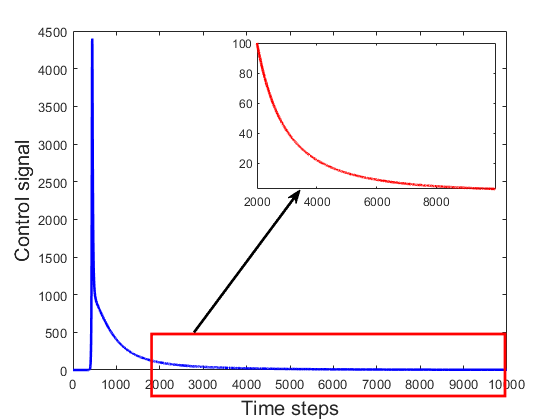}  
		\caption{}
		\label{fig:sub-second}
	\end{subfigure}
	\caption{(\ref{fig:sub-first} ): Time evolution of the populations  of the $\ket{0}\bra{0}$ and $\ket{1}\bra{1}$ states. (\ref{fig:sub-second}): Time evolution of  the control signal  $u_t$ responsible of achieving the control objective.}
	\label{figure_1}
\end{figure}
\begin{figure}[h!]

	\centering
	% include second image
	\includegraphics[width=.5\linewidth]{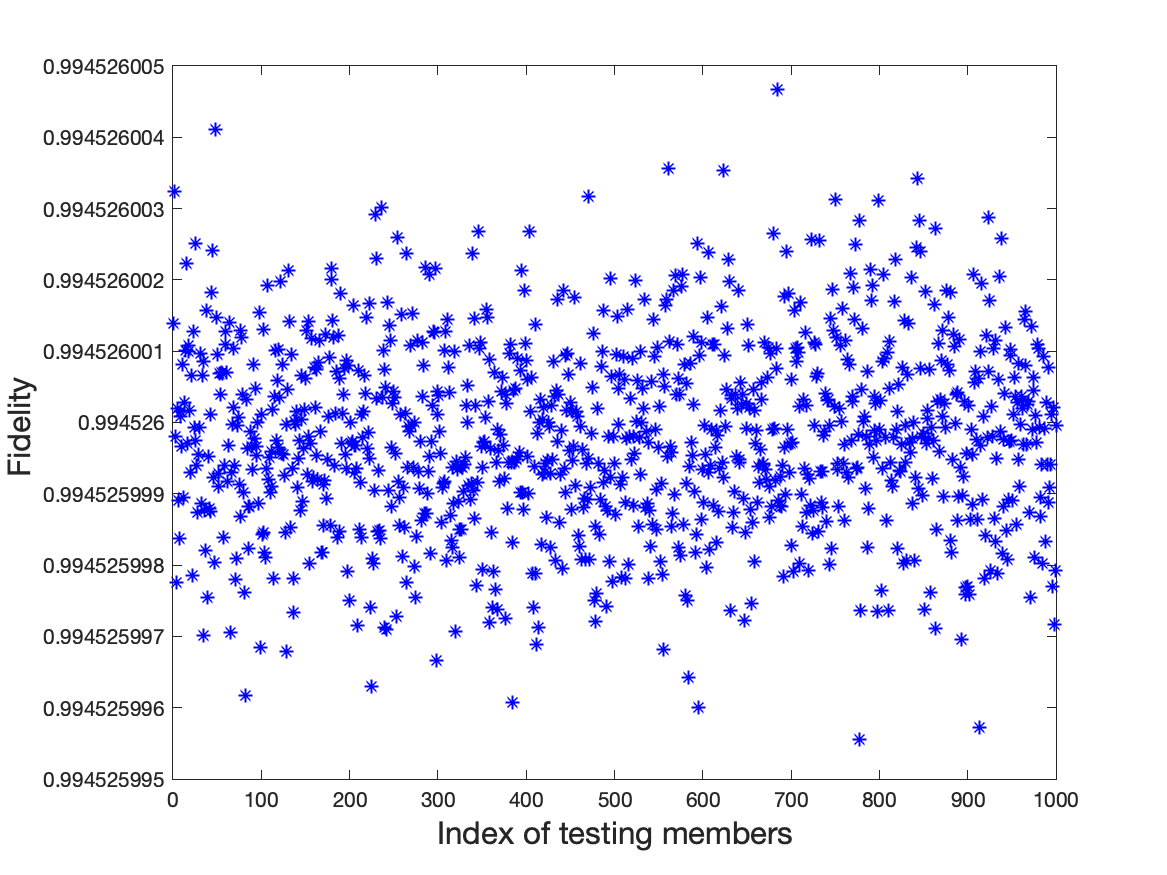}
	\caption{ The testing performance of the  optimal control for spin-$1/2$ ensemble. The samples are randomly generated through the noise term $\eta(x_{t-1})$. }
			\label{performance}  
\end{figure}
%%%%%%%%%%%%%%%%%%%%%%%%%%%%%%%%%%%
\subsubsection{$\Lambda$-type atomic ensemle}\label{atomicsystem}
\begin{figure}[t!]
	\begin{subfigure}{.5\textwidth}
		\centering
		% include first image
		\includegraphics[width=.8\linewidth]{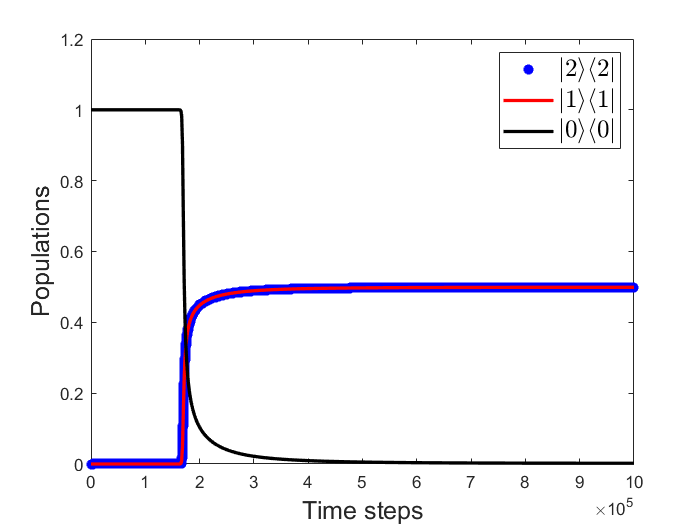}  
		\caption{}
		\label{figure4}
	\end{subfigure}
	\begin{subfigure}{.5\textwidth}
		\centering
		% include second image
		\includegraphics[width=.8\linewidth]{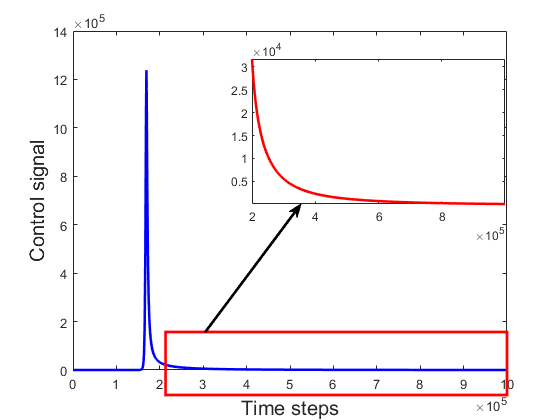}  
		\caption{}
		\label{figure5}
	\end{subfigure}
	\caption{(\ref{figure4} ):  Time evolution of the populations the $\ket{0}\bra{0}$,  $\ket{1}\bra{1}$ and $\ket{2}\bra{2}$ states. (\ref{figure5}): Time evolution of  the control signal,  $u_t$ responsible of achieving the control objective.}
	\label{figure_3}
\end{figure}
\begin{figure}[h!]
	
	\centering
	% include second image
	\includegraphics[width=.5\linewidth]{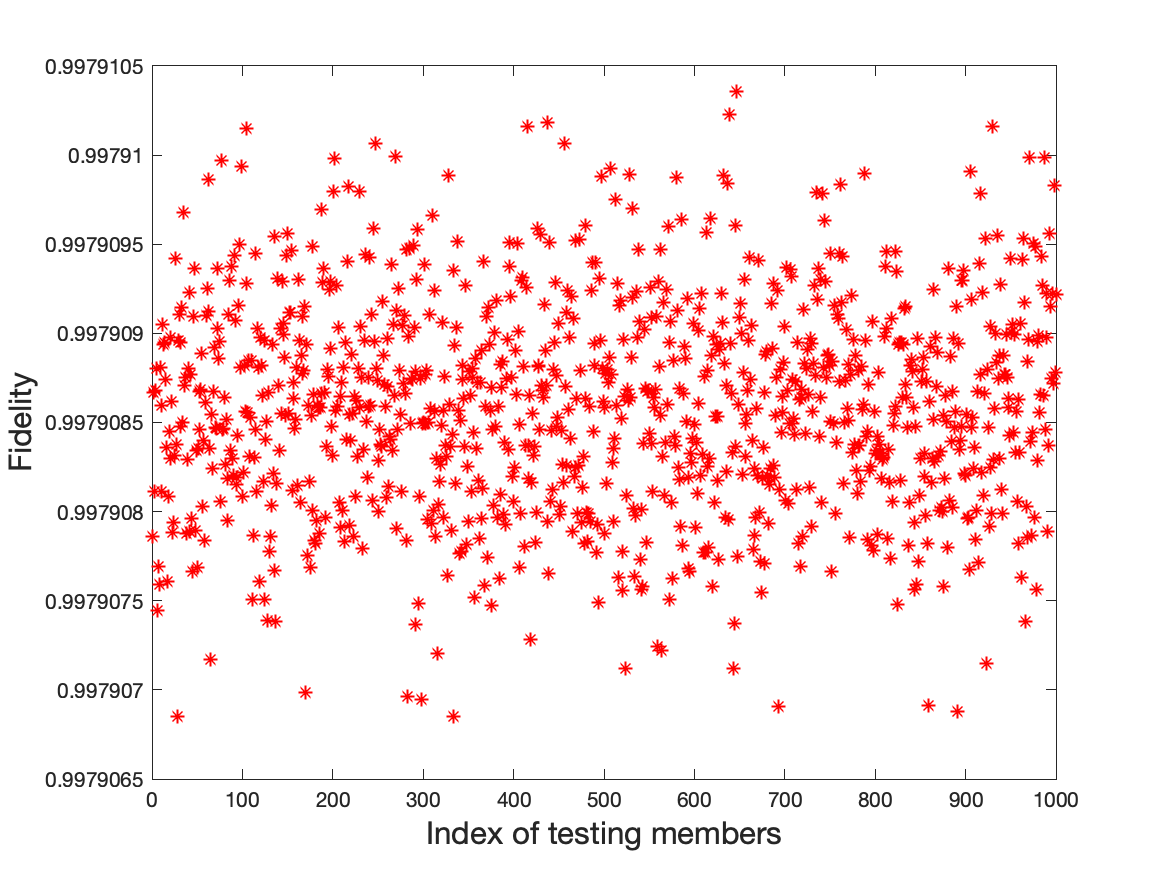}
	\caption{ The testing performance of the  optimal control for three level systems ensemble. The samples are randomly generated through the noise term $\eta(x_{t-1})$. }
		\label{performance2}  
\end{figure}

We now investigate a $\Lambda$-type atomic system \cite{spin} prepared in an initial state and interacts with an electric field ${u}_{t}$. The Hamiltonian describing this interaction is represented by matrices in the basis ${\ket{2},\ket{1},\ket{0}}$ as follows:

\begin{equation}
H=H_0+{u}_{t}H_1=\left(\begin{array}{ccc}
\dfrac{3}{2} & 0 & 0 \\
0 & 1 & 0 \\
0 & 0&0 \\
\end{array}\right)+\left(\begin{array}{ccc}
0 & 0 & 1 \\
0 & 0 & 1 \\
1 & 1&0 \\
\end{array}\right) {u}_{t}.
\end{equation}
Moreover, we assume that the system interacts with an environment, with the system-environment coupling being described by a single Lindblad operator, $L_{02}=\sqrt{\Theta}\ket{0}\bra{2}$, where $\Theta\equiv\Theta_{2\to 0}$ represents the dissipative transition rate from Hamiltonian eigenstate $\ket{2}$ to the Hamiltonian ground state $\ket{0}$.

The control objective aims to transfer the system from state $\ket{0}\bra{0}$ to the desired state $\rho_d=\ket{\psi_d}\bra{\psi_d}$, with $\ket{\psi_d}=\dfrac{1}{\sqrt{2}}(\ket{1}+\ket{2})$. To achieve this, we must maximize the fidelity between the actual state and the target state $\rho_d$. The target operator $D$ in (\ref{outp}) now takes the form $D=(\text{vec}(\Pi_d)^T)^T$, where $\Pi_d=\ket{\psi_d}\hspace{0.05cm}\bra{\psi_d}$ represents the projector onto the target state $\ket{\psi_d}$. By focusing on this control objective, we can ensure a high degree of precision in the transfer of the system between states, even when accounting for the environmental interactions. This approach highlights the versatility of the probabilistic quantum control method in effectively managing different quantum systems and control objectives.

The elements of the $\tilde{A}$ and $\tilde{N}$ matrices are provided in Appendix (\ref{spinn_1}). These matrices are then employed to compute the time-evolution of matrices $A$ and $B$ using (\ref{A_B_t}) at each time step. With these matrices and setting $x_e=(0,0,0,0,0,0,0,0,0)^T$, $\Delta{t}=0.0000000001$, $G_r=0.0000001$, $\Omega=10000$, and $\Theta=0.9$, we evaluate the matrices $M_t$ and $P_t$ defined in equations (\ref{Mt}) and (\ref{Pt}), respectively, at each instant of time, as discussed in the optimization phase of Algorithm (\ref{alg:cap}). The evolution of $u_t$ is then obtained by substituting the calculated $M_t$ and $P_t$ matrices into (\ref{optimalcapp}). We repeat these steps until the fidelity between the actual state and the target state reaches its maximum value of unity, signifying that the system's state has reached the desired target state.

Following the optimization phase, the obtained optimal randomized controller is used to control one of the members of the $\Lambda-$type system. Figure (\ref{figure4}) illustrates the time evolution of the populations of states $\ket{0}\bra{0}$, $\ket{1}\bra{1}$, and $\ket{2}\bra{2}$ of the selected ensemble  interacting with the calculated electric field. The time evolution of the obtained optimal control signal is depicted in Figure (\ref{figure5}). It is evident that the optimal control objective is achieved after only a few steps. As a result, the optimal control is capable of preserving quantum coherence, i.e., the quantum superposition, under the influence of decoherence.

To further demonstrate the capability of the designed randomised controller to effectively control the members of the $\Lambda-$type system,  the optimal controller is applied to other randomly selected ensemble members generated through the noise term $\eta(x_{t-1})$. These ensemble members represent different realizations of the state time evolution influenced by the noise term $\eta(x_{t-1})$. Figure (\ref{performance2}) displays the fidelities between the target state and the final states of $1000$ ensemble members interacting with the optimal controller. We observe that the fidelities are almost equal to 1, implying that all members successfully reached the target state $\ket{\psi_d}$. This result highlights the potential effectiveness of the proposed method in controlling $\Lambda$-type atomic systems under environmental effects.

\section{Final comments}\label{conclu}
 In this paper, we have presented a novel approach for deriving optimal randomized controllers for quantum ensembles. Our proposed method consists of two main phases: optimization and testing. The optimization phase involves constructing a control for a generalized system of the ensemble using a probabilistic approach, which is based on describing the system dynamics with probability density functions and minimizing the distance between the actual closed-loop joint pdf and the desired joint pdf. We first established a general control solution for quantum systems described by arbitrary pdfs. The solution is then demonstrated on a specific case involving quantum systems affected by Gaussian noise and described by Gaussian pdfs, which yielded an analytical form for the randomized controller.

In the testing phase, the controller obtained during the optimization phase is applied to a large number of randomly selected ensemble members. Our results demonstrated that, with the optimized controller, all members of the quantum ensemble that were also affected by environmental noises and interactions were successfully steered to the same target state, illustrating the effectiveness of our proposed methodology. In practical applications, the verification of successful steering to the target state would involve performing quantum state tomography on the final state of the system, and calculating the fidelity between the actual state and the desired state. Moreover, it's important to note a crucial aspect of quantum mechanics: the act of measurement itself can influence the state of the quantum system, a phenomenon known as measurement backaction. While this effect was not explicitly accounted for in the current work, it's an important consideration as it can potentially impact the effectiveness of the control field. Future work could explore integrating the effect of measurement backaction into the control methodology to ensure robust control even in the presence of frequent measurements.

In practical settings, the approach entails deriving the control field using historical, simulated, or real-time data obtained from quantum systems with the introduced method, and subsequently applying the electric field to the quantum ensemble.

The main advantages of the proposed methodology include:
\begin{itemize}
\item It provides a general form for the time evolution of the controller for systems affected by arbitrary noise, making it suitable for a wide range of quantum systems.

\item It is easy to implement numerically, as explained in Algorithm (\ref{alg:cap}), and does not require significant computational resources.

\item The method is highly adaptable and can be applied to various quantum systems and control objectives, making it appealing to physicists and experimentalists working with quantum systems.
\end{itemize}

Our approach allows for the efficient control and manipulation of quantum systems, even in the presence of environmental noises and interactions. This can be particularly useful for preserving quantum coherence and achieving specific control objectives in various experimental settings.

Possible extensions of this work include considering quantum systems with arbitrary forms of noise affecting their evolutions. This would further expand the applicability and versatility of our proposed method, enabling it to address a broader range of quantum systems and control challenges. 

%Other future work may involve:
%
%Investigating the robustness of the proposed method against uncertainties in system parameters or environmental conditions, which are common in experimental scenarios.
%
%Adapting the method to handle time-varying noise or situations where the noise statistics change during the system's evolution.
%
%Exploring the application of the proposed method to quantum error correction, which is crucial for the development of fault-tolerant quantum computing.
%
%Implementing the proposed method in experimental setups to validate its performance in real-world quantum systems and assess its practical feasibility.

\section*{Acknowledgements:} This work was supported by the EPSRC grant EP/V048074/1.
\section*{Appendices}

\appendix
\renewcommand{\theequation}{A.\arabic{equation}}
\setcounter{equation}{0} % ap
\section{Vectorization of the density operator}\label{vect_app}
In equation (\ref{LVN_eq}), we consider a density matrix $\rho(t)\in\mathbf{C}^{l \times l}$ given by:
\begin{equation}
	\rho(t)=(\rho(t))^\dagger=\left(
	\begin{array}{ccccccccc}
		\rho_{0,0}(t)&\rho_{0,1}(t) & \dots &\rho_{0,l-1}(t) \\
		\rho_{1,0}(t)&\rho_{1,1}(t) & \dots &\rho_{1,l-1}(t) \\
		\vdots&\vdots & \ddots &\vdots \\
		\rho_{l-1,0}(t)&\rho_{l-1,1}(t) & \dots &\rho_{l-1,l-1}(t) \\
	\end{array}
	\right).\end{equation}
The vectorization of $\rho(t)$ used in (\ref{NLVN_pa}) is:
\begin{align}\label{x_t_vet}
	\tilde{x}(t)&=\text{vec}(\rho(t)) \nonumber \\&=\begin{bmatrix}
		\rho_{0,0}(t) & \rho_{1,1}(t) 
		\dots
		\rho_{l-1,l-1}(t)& \rho_{0,1}(t) \dots  
		\rho_{0,l-1}(t)&\rho_{1,0}(t) 
		\dots\rho_{l-1,0}(t) 
		\dots\dots\dots\rho_{l-2,l-1}(t)&\rho_{l-1,l-2}(t) 
	\end{bmatrix}^T,
\end{align}
where $T$ stands for the transpose operation.
Thus, for $l=3$ the vectorization of the density matrix $\rho(t)	=\left(
\begin{array}{ccc}
	\rho_{00}(t)&\rho_{01}(t) & \rho_{02}(t)\\
	\rho_{10}(t) & \rho_{11}(t) & \rho_{12}(t) \\
	\rho_{20}(t) & \rho_{21}(t) & \rho_{22}(t) 
\end{array}\right)$ is trivially given by:
\begin{equation}
\tilde{	x}(t)=\tilde{	x}_t=
	\begin{bmatrix}
		\rho_{00}(t) & \rho_{11}(t) & \rho_{22}(t) &
		\rho_{01} (t) &  \rho_{02}(t)  & \rho_{10} (t) &  \rho_{20}(t)  & \rho_{12}(t) &  \rho_{21}(t) \end{bmatrix}^T.
\end{equation}

\section{Complex Gaussian  and Gaussian distributions}\label{app_C}
\setcounter{equation}{0} % ap
\renewcommand{\theequation}{B.\arabic{equation}}
For nonsingular covariance matrix  $\Gamma_t$, the complex Gaussian distribution of a complex random variable $x_t\in \mathbf{C}^n$ is given by:
\begin{align}
	&\mathcal{N}_{\mathcal{C}}(\mu_{t},{\Gamma_t})=\dfrac{1}{\pi^n|\Gamma_t|} \exp\bigg[- {({{x}_t}- {{\mu} _t})^{\dagger}}{\Gamma_t}^{-1}({x_t} - {\mu _t})\bigg],
\end{align}
where $\mu_t=E(x_t)$, and $\Gamma_t=\text{E}((x_t-\mu_t)(x_t-\mu_t)^\dagger)$, with $\text{E}(x_t)$ being the expected value of $x_t$, $|\Gamma_t|$ is the determinant of $\Gamma_t$, and $\dagger$ stands for the complex conjugate transpose operation.\\
The  Gaussian distribution for a real random variable $o_t\in \mathbf{R}^n$ is given by:
\begin{align}\label{NGaussian}
	\mathcal{N}(o_{m},{G})&=\dfrac{1}{\sqrt{(2\pi)^n|G|}} \exp\bigg[- {0.5({{o}_t}- {{o} _m})^{T}}{G}^{-1}({o_t} - {o _m})\bigg].
%	\nonumber \\
%	&\equiv \dfrac{1}{\sqrt{(2\pi)^n|G|}} \exp\bigg[- {0.5({{o}_t}- {{o} _m})^{\dagger}}{G}^{-1}({o_t} - {o _m})\bigg],
\end{align}
where $o_m=E(o_t)$ and  $G=\text{E}((o_t-o_m)(o_t-o_m)^T)$ are respectively the mean and covariance matrices, and $T$ stands for the transpose operation.  Since $o_t$ belongs to the set of real numbers $\mathbf{R}^n$, the Gaussian distribution (\ref{NGaussian}) can be equivalently restructured as:
\begin{align}
\mathcal{N}(o_{m},{G})&=\dfrac{1}{\sqrt{(2\pi)^n|G|}} \exp\bigg[- {0.5({{o}_t}- {{o} _m})^{\dagger}}{G}^{-1}({o_t} - {o _m})\bigg].
\end{align}
This recasting of the real Gaussian distribution into the form of a complex Gaussian distribution provides a consistency in the derivation of the randomized controller. It not only simplifies the presentation but also aligns well with the inherent nature of quantum systems, where the state variable is complex. Thus, the adapted representation facilitates a more coherent theoretical framework for our probabilistic quantum control approach.
\section{Evaluation of the optimal cost-to-go, $\gamma(x_{t-1})$}\label{AppA}
\setcounter{equation}{0} % ap
\renewcommand{\theequation}{C.\arabic{equation}}
%In this section, we aim to calculate The from of the performance index function  given in Eq. (\ref{49}). We first   evaluate of the coefficient $\beta(u_{t-1},x_{t-1})$ defined in equation (\ref{beta}) and repeated here,
In this appendix we show how to obtain the optimal cost-to-go function (\ref{49}) stated in Theorem (\ref{Theo1}). For this purpose, we first evaluate the coefficient $\beta(u_{t-1},x_{t-1})$ defined in equation (\ref{beta}), repeated here:
	\begin{align}\label{betaapp1}
		\beta(u_{t-1},x_{t-1}) = \int s(x_{t}|u_{t-1}, x_{t-1}) s(o_t| x_{t}) \times\ln \bigg ( \frac{ s(o_t| x_{t})}{^Is(o_t| x_{t})} \dfrac{1}{ {\gamma}(x_{t})}\bigg) \mathrm{d} x_{t}\mathrm{d} o_t.
	\end{align}
Using  equations (\ref{eq:eq5}),  (\ref{eq:eq2}) and (\ref{49}) we evaluate:
	\begin{align}\label{betaapp2}
		& \ln \bigg ( \frac{ s(o_t| x_{t})}{^Is(o_t| x_{t})} \dfrac{1}{ {\gamma}(x_{t})}\bigg)\nonumber\\
		& 
		=- 0.5{({o_t} - D{x_t})^{\dagger}}{G ^{ - 1}}({o_t} - D{x_t})+ 0.5{({o_t} - {{o}_{d}})^{\dagger}}G_r^{ - 1}({o_t} - {{o}_{d}}) +0.5x_{t}^{\dagger}{M_{t}}{x_{t}} +0.5P_{t}{x_{t}} + 0.5{\omega_{t}}+0.5\ln\bigg(\dfrac{|G_r|}{|G|}\bigg),	\end{align}
which implies that:
\begin{align}\label{betaapp2_}
		& \ln \bigg ( \frac{ s(o_t| x_{t})}{^Is(o_t| x_{t})} \dfrac{1}{ {\gamma}(x_{t})}\bigg)\nonumber\\
& =  - 0.5o_t^{\dagger}{G ^{ - 1}}{o_t} - 0.5x_t^{\dagger}{D^{\dagger}}{G ^{ - 1}}D{x_t} + o_t^{\dagger}{G ^{ - 1}}D{x_t}  + 0.5o_t^{\dagger} G_r^{ - 1}{o_t}- o_t^{\dagger} G_r^{ - 1}{{o}_{d}} +0.5x_{t}^{\dagger} {M_{t}}{x_{t}} +0.5P_{t}{x_{t}} +0.5 {\omega_{t}}\nonumber\\
		& + 0.5{o}_{d}^{\dagger} G_r^{ - 1}{{o}_{d}}+0.5\ln\bigg(\dfrac{|G_r|}{|G|}\bigg)  \nonumber\\
		&=  - 0.5o_t^{\dagger} ({G ^{ - 1}} - G_r^{ - 1}){o_t} + o_t^{\dagger} ({G ^{ - 1}}D{x_t} - G_r^{ - 1}{{o}_{d}}) -0.5 x_t^{\dagger} D^{\dagger} {G ^{ - 1}}D{x_t} +0.5x_{t}^{\dagger} {M_{t}}{x_{t}} +0.5P_{t}{x_{t}} + 0.5{\omega_{t}} \nonumber\\
		&+ 0.5{o}_{d}^{\dagger} G_r^{ - 1}{{o}_{d}}+0.5\ln\bigg(\dfrac{|G_r|}{|G|}\bigg).
	\end{align}
By substituting (\ref{betaapp2_})  into (\ref{betaapp1}), and integrating over  $o_t$ we get:
	\begin{align}\label{betaapp3}
		&\beta(u_{t-1},x_{t-1})=\int s(x_{t}|u_{t-1}, x_{t-1}) s(o_t| x_{t})\bigg( - 0.5o_t^{\dagger} ({G ^{ - 1}} - G_r^{ - 1}){o_t} + o_t^{\dagger} ({G ^{ - 1}}D{x_t} - G_r^{ - 1}{{o}_{d}}) -0.5 x_t^{\dagger} D^{\dagger} {G ^{ - 1}}D{x_t}\nonumber\\
		&  +0.5x_{t}^{\dagger} {M_{t}}{x_{t}}+0.5P_{t}{x_{t}} + 0.5{\omega_{t}}+ 0.5{o}_{d}^{\dagger} G_r^{ - 1}{{o}_{d}} +0.5\ln\bigg(\dfrac{|G_r|}{|G|})\bigg) \mathrm{d} x_{t}\mathrm{d} o_t\nonumber\\
		& = \int {s\left( {{x_t}|{u_{t-1}},{\rm{ }}{x_{t - 1}}} \right){\rm{ }}} \bigg( 0.5x_t^{\dagger} {D^{\dagger} }G_r^{ - 1}D{x_t} - x_t^{\dagger} {D^{\dagger} }G_r^{ - 1}{{o}_{d}}+0.5x_{t}^{\dagger} {M_{t}}{x_{t}} +0.5P_{t}{x_{t}} + 0.5{\omega_{t}+}0.5{o}_{d}^{\dagger} G_r^{ - 1}{{o}_{d}} \nonumber\\
		&-0.5 \Tr((G^{-1}-G_r^{ - 1})G) +0.5\ln\bigg(\dfrac{|G_r|}{|G|}\bigg)\bigg)d{x_t} \nonumber\\
		& = \int {s\left( {{x_t}|{u_{t-1}},{\rm{ }}{x_{t - 1}}} \right){\rm{ }}} \bigg( 0.5x_t^{\dagger} ({D^{\dagger} }G_r^{ - 1}D+M_{t}){x_t}+ 0.5(P_{t}-2{o}_{d}^{\dagger} G_r^{ - 1}{D})x_t +0.5{\omega_{t}} + 0.5{o}_{d}^{\dagger} G_r^{ - 1}{{o}_{d}} +0.5\ln\bigg(\dfrac{|G_r|}{|G|}\bigg)  \nonumber\\
		&-0.5 \Tr((G^{-1}-G_r^{ - 1})G)\bigg)d{x_t}.
	\end{align}
Integrating over $x_t$ yields the following form of $\beta(u_{t-1},x_{t-1})$:
	\begin{align} \label{betaapp4}
		&\beta(u_{t-1},x_{t-1})= 0.5\mu_t^{\dagger} ({D^{\dagger} }G_r^{ - 1}D+M_{t}){\mu_t} +  0.5(P_{t}-2{o}_{d}^{\dagger} G_r^{ - 1}{D})\mu_t+ 0.5{\omega_{t}} + 0.5{o}_{d}^{\dagger} G_r^{ - 1}{{o}_{d}}  +0.5\ln\bigg(\dfrac{|G_r|}{|G|}\bigg)\nonumber\\
		&-0.5 \Tr((G^{-1}-G_r^{ - 1})G)+0.5\Tr(({D^{\dagger} }G_r^{ - 1}D+M_{t})\Gamma_t).
	\end{align}
Using the definition of $\Gamma_t$ given in equation (\ref{eq:InDepCov}), the last term in the above equation can be evaluated to give:
\begin{align}
\Tr(({D^{\dagger} }G_r^{ - 1}D+M_{t})\Gamma_t)&=\Tr(({D^{\dagger} }G_r^{ - 1}D+M_{t}) A{x}_{t-1} \Sigma x_{t-1}^{\dagger}{A}^{\dagger}), \nonumber \\
&=<\text{cyclic permutation}> \nonumber \\
& \Tr({A}^{\dagger}({D^{\dagger} }G_r^{ - 1}D+M_{t}) A{x}_{t-1} \Sigma x_{t-1}^{\dagger}), \nonumber \\
&=<\text{identity $\Tr(Q z z^{\dagger})= z^{\dagger} Q z$}> \nonumber \\
& x_{t-1}^{\dagger} {A}^{\dagger}({D^{\dagger} }G_r^{ - 1}D+M_{t}) A{x}_{t-1} \Sigma. 
\end{align}
Remembering that $\mu_t=Ax_{t-1}+Bu_{t-1}$, the relation (\ref{betaapp4}) becomes:
	\begin{align}\label{betaapp5}
	&\beta(u_{t-1},x_{t-1})= 0.5(Ax_{t-1}+Bu_{t-1})^{\dagger} ({D^{\dagger} }G_r^{ - 1}D+M_{t}){(Ax_{t-1}+Bu_{t-1})} +0.5(P_{t}-2{{o}_{d}^{\dagger} }G_r^{ - 1}D)  (Ax_{t-1}+Bu_{t-1}) \nonumber\\& 
+	0.5{\omega_{t}} + 0.5{o}_{d}^{\dagger} G_r^{ - 1}{{o}_{d}}  +0.5\ln\bigg(\dfrac{|G_r|}{|G|}\bigg)-0.5 \Tr((G^{-1}-G_r^{ - 1})G)+0.5x_{t-1}^{\dagger}{A}^{\dagger}({D^{\dagger} }G_r^{ - 1}D+M_{t}){\Sigma}A{x}_{t-1}
	\nonumber\\	&= 0.5x_{t-1}^{\dagger} A^{\dagger} ({D^{\dagger} }G_r^{ - 1}D+M_{t}){Ax_{t-1}} +0.5(P_{t}-2{{o}_{d}^{\dagger} }G_r^{ - 1}D) Ax_{t-1} + 0.5 u_{t-1}^{\dagger} B^{\dagger}({D^{\dagger} }G_r^{ - 1}D+M_{t}){Bu_{t-1}} \nonumber\\
	&+ u_{t-1}^{\dagger} B^{\dagger} ({D^{\dagger} }G_r^{ - 1}D+M_{t}){Ax_{t-1}}+0.5(P_{t}-2{{o}_{d}^{\dagger} }G_r^{ - 1}D) B u_{t-1} + 
	0.5{\omega_{t}} + 0.5{o}_{d}^{\dagger} G_r^{ - 1}{{o}_{d}}  +0.5\ln\bigg(\dfrac{|G_r|}{|G|}\bigg)\nonumber\\&-0.5 \Tr((G^{-1}-G_r^{ - 1})G)+0.5x_{t-1}^{\dagger}{A}^{\dagger}({D^{\dagger} }G_r^{ - 1}D+M_{t}){\Sigma }A{x}_{t-1}
\end{align}
%\section{Calculation of the performance index $\gamma(x_{t-1})$}\label{AppB}
%\setcounter{equation}{0} % ap
%\renewcommand{\theequation}{B.\arabic{equation}}
By substituting the form of $\beta$ found in (\ref{betaapp5}) and the ideal distribution of the controller provided in equation (\ref{eq:eq3}), into the definition of $\gamma(x_{t-1})$ given in  (\ref{gamma2}) we find that:
	\begin{align}\label{gamma}
		&\gamma(x_{t-1}) = \int {^Ic(u_{t-1}|x_{t-1}) \exp[{-\beta(u_{t-1},x_{t-1})}] \mathrm{d} u_{t-1}}, \nonumber \\
		& =(2\pi )^{-1/2}|\Omega|^{-1/2}\int {\exp } \bigg[ - 0.5{({u_{t-1}} - {u_r})^{\dagger} }{\Omega ^{ - 1}}({u_{t-1}} - {u_r}) -0.5x_{t-1}^{\dagger} A^{\dagger} ({D^{\dagger} }G_r^{ - 1}D+M_{t}){Ax_{t-1}} \nonumber\\
		& -0.5(P_{t}-2{{o}_{d}^{\dagger} }G_r^{ - 1}D) Ax_{t-1} - 0.5u_{t-1}^{\dagger} B^{\dagger} ({D^{\dagger} }G_r^{ - 1}D+M_{t}){Bu_{t-1}}-u_{t-1}^{\dagger} B^{\dagger} ({D^{\dagger} }G_r^{ - 1}D+M_{t}){Ax_{t-1}}\nonumber\\& -0.5(P_{t}-2{{o}_{d}^{\dagger} }G_r^{ - 1}D) B u_{t-1}  - 0.5{\omega_{t}} - 0.5{o}_{d}^{\dagger} G_r^{ - 1}{{o}_{d}}  -0.5\ln\bigg(\dfrac{|G_r|}{|G|}\bigg)+0.5 \Tr((G^{-1}-G_r^{ - 1})G) \nonumber \\ &-0.5x_{t-1}^{\dagger}{A}^{\dagger}({D^{\dagger} }G_r^{ - 1}D+M_{t}){\Sigma }A{x}_{t-1}\bigg] {\rm{d}}{u_{t-1}} \nonumber\\
		&=(2\pi)^{-1/2}|\Omega|^{-1/2} {\exp } \bigg(-0.5x_{t-1}^{\dagger} A^{\dagger} ({D^{\dagger} }G_r^{ - 1}D+M_{t}){Ax_{t-1}} -0.5(P_{t}-2{{o}_{d}^{\dagger} }G_r^{ - 1}D) Ax_{t-1} 
		- 0.5{\omega_{t}}- 0.5{o}_{d}^{\dagger} G_r^{ - 1}{{o}_{d}}  \nonumber\\& -0.5\ln\bigg(\dfrac{|G_r|}{|G|}\bigg)+0.5 \Tr((G^{-1}-G_r^{ - 1})G)-0.5x_{t-1}^{\dagger}{A}^{\dagger}({D^{\dagger} }G_r^{ - 1}D+M_{t}){\Sigma }A{x}_{t-1}
			-0.5u_r^{\dagger} \Omega^{-1}u_r\bigg)\nonumber\\
		&\int {\exp } \bigg(-0.5u_{t-1}^{\dagger} (\Omega^{-1}+B^{\dagger} ({D^{\dagger} }G_r^{ - 1}D+M_{t})B){u_{t-1}}+u_{t-1}^{\dagger} \big(\Omega^{-1}u_r-B^{\dagger} ({D^{\dagger} }G_r^{ - 1}D+M_{t}){Ax_{t-1}}\nonumber\\
		&-0.5B^{\dagger} (P_{t}-2{{o}_{d}^{\dagger} }G_r^{ - 1}D)^{\dagger} \big) \bigg) {\rm{d}}{u_{t-1}} 
			\end{align} 
By completing the square with respect to $u_{t-1}$ in the above equation, and using the solution of the general multiple integral given in Theorem (10.5.1) in Ref \cite{Graybill}, it follows that:
	\begin{align}\label{gamma_1}
&\gamma(x_{t-1})= {\exp } \bigg(-0.5x_{t-1}^{\dagger} A^{\dagger} ({D^{\dagger} }G_r^{ - 1}D+M_{t}){Ax_{t-1}} -0.5(P_{t}-2{{o}_{d}^{\dagger} }G_r^{ - 1}D) Ax_{t-1} - 0.5{\omega_{t}}- 0.5{o}_{d}^{\dagger} G_r^{ - 1}{{o}_{d}}  \nonumber\\& -0.5\ln\bigg(\dfrac{|G_r|}{|G|}\bigg)+0.5 \Tr((G^{-1}-G_r^{ - 1}))-0.5x_{t-1}^{\dagger}{A}^{\dagger}({D^{\dagger} }G_r^{ - 1}D+M_{t}){\Sigma }A{x}_{t-1}
-0.5u_r^{\dagger} \Omega^{-1}u_r\bigg) \nonumber\\
&\times{\exp }\bigg(0.5 \big(\Omega^{-1}u_r-B^{\dagger} ({D^{\dagger} }G_r^{ - 1}D+M_{t}){Ax_{t-1}}-0.5B^{\dagger} (P_{t}-2{{o}_{d}^{\dagger} }G_r^{ - 1}D)^{\dagger} \big)^{\dagger} \nonumber \\ 
&\big(\Omega^{-1}+B^{\dagger} ({D^{\dagger} }G_r^{ - 1}D+M_{t})B\big)^{-1} 
\big(\Omega^{-1}u_r-B^{\dagger} ({D^{\dagger} }G_r^{ - 1}D+M_{t}){Ax_{t-1}}-0.5B^{\dagger} (P_{t}-2{{o}_{d}^{\dagger} }G_r^{ - 1}D)^{\dagger} \big) \nonumber \\&
 \bigg)\times{|\Omega|^{-1/2}}|\Omega^{-1}+B^{\dagger} ({D^{\dagger} }G_r^{ - 1}D+M_{t})B|^{-1/2}
	\end{align}
Finally the form provided in (\ref{49}) in Theorem (\ref{Theo1}) is found:
	\begin{align}\label{appendb2}
		&-\ln{(\gamma(x_{t-1}))}\nonumber\\
		&=  0.5x_{t-1}^{\dagger} \bigg(A^{\dagger} (D^{\dagger} G_r^{-1}D+M_{t})A-A^{\dagger} (D^{\dagger} G_r^{-1}D+M_{t})^{\dagger} B
		\big(\Omega^{-1}+B^{\dagger} ({D^{\dagger} }G_r^{ - 1}D+M_{t})B\big)^{-1}
		\nonumber\\
	&	B^{\dagger} (D^{\dagger} G_r^{-1}D+M_{t})A+{A}^{\dagger}({D^{\dagger} }G_r^{ - 1}D+M_{t}){\Sigma }A\bigg)x_{t-1}
	\nonumber\\
&	+0.5\bigg((P_{t}-2{o}_{d}^{\dagger} G_r^{-1}D)A+2(\Omega^{-1}u_r-0.5B^{\dagger} (P_{t}^{\dagger} -2D^{\dagger} G_r^{-1}{o}_{d}))^{\dagger} 
		\big(\Omega^{-1}+B^{\dagger} ({D^{\dagger} }G_r^{ - 1}D+M_{t})B\big)^{-1}\nonumber\\
		&B^{\dagger} (D^{\dagger} G_r^{-1}D+M_{t})A
	\bigg)	x_{t-1}	\nonumber\\
	&+0.5\bigg({\omega_{t}}+ {o}_{d}^{\dagger} G_r^{ - 1}{{o}_{d}}   +\ln\bigg(\dfrac{|G_r|}{|G|}\bigg)- \Tr\big((G^{-1}-G_r^{ - 1})G\big)
	+u_r^{\dagger} \Omega^{-1}u_r	\nonumber\\
	&-\big( \Omega^{-1}u_r-0.5B^{\dagger} (P_{t}^{\dagger} -2{D}^{\dagger}G_r^{ - 1}{{o}_{d}})\big)^{\dagger} 
	\big(\Omega^{-1}+B^{\dagger} ({D^{\dagger} }G_r^{ - 1}D+M_{t})B\big)^{-1}	
	\big( \Omega^{-1}u_r-0.5B^{\dagger} (P_{t}^{\dagger} -2{D}^{\dagger}G_r^{ - 1}{{o}_{d}})\big)\nonumber\\
	&-2\ln({|\Omega|^{-1/2}}|\Omega^{-1}+B^{\dagger} ({D^{\dagger} }G_r^{ - 1}D+M_{t})B|^{-1/2})\bigg).
	\end{align}
\section{Calculating of the control distribution function}\label{appC}
\setcounter{equation}{0} % ap
\renewcommand{\theequation}{D.\arabic{equation}}
By substituting (\ref{eq:eq3}), (\ref{betaapp5}) and (\ref{appendb2}) into  (\ref{eq:Eqn4}),
it follows that the optimal control distribution is given by:
\begin{align}
 &c(u_{t-1}|x_{t-1})=(2\pi)^{-1/2}|\Omega^{-1}+B^{\dagger} ({D^{\dagger} }G_r^{ - 1}D+M_{t})B|^{1/2}\exp\bigg[-0.5u_{t-1}^{\dagger}\bigg(\Omega^{-1}+B^{\dagger} ({D^{\dagger} }G_r^{ - 1}D+M_{t})B\bigg){u_{t-1}}\nonumber\\
 &+u_{t-1}^{\dagger} \bigg(\Omega^{-1}u_r-B^{\dagger} ({D^{\dagger} }G_r^{ - 1}D+M_{t}){Ax_{t-1}}
 -0.5B^{\dagger} (P_{t}-2{{o}_{d}^{\dagger} }G_r^{ - 1}D)^{\dagger} \bigg) 
-0.5\bigg({\Omega ^{ - 1}}{u_r} - {B^{\dagger} }({D^{\dagger} }G_r^{ - 1}D + {M_t})A{x_{t -1}}\nonumber\\
& -0.5 {B^{\dagger} }( P_{t }^{\dagger}  - 2{D^{\dagger} }G_r^{ - 1}{{o}_{d}} )\bigg)^{\dagger}\bigg({\Omega ^{ - 1}} + {B^{\dagger} }( {D^{\dagger} }G_r^{ - 1}D + {M_t})B\bigg)^{-1} \bigg({\Omega ^{ - 1}}{u_r} - {B^{\dagger} }({D^{\dagger} }G_r^{ - 1}D + {M_t})A{x_{t -1}}\nonumber\\
& -0.5 {B^{\dagger} }( P_{t }^{\dagger}  - 2{D^{\dagger} }G_r^{ - 1}{{o}_{d}} )\bigg)
\bigg].
\end{align}
which can be written as:
\begin{align}
	c(u_{t-1}|x_{t-1})&=(2\pi)^{-1/2}|\Omega^{-1}+B^{\dagger} ({D^{\dagger} }G_r^{ - 1}D+M_{t})B|^{1/2}\\\nonumber& \times\exp\bigg[- {0.5({{u}_{t-1}}- {{v} _{t-1}})^{\dagger}}\bigg(\Omega^{-1}+B^{\dagger} ({D^{\dagger} }G_r^{ - 1}D+M_{t})B\bigg)({u_{t-1}} - {v _{t-1}})\bigg],
\end{align}
where,
\begin{align}\label{optimalcappc}
	v_{t-1}&=\bigg({\Omega ^{ - 1}} + {B^{\dagger} }( {D^{\dagger} }G_r^{ - 1}D + {M_t})B\bigg)^{-1} \bigg({\Omega ^{ - 1}}{u_r} - {B^{\dagger} }({D^{\dagger} }G_r^{ - 1}D + {M_t})A{x_{t -1}} -0.5 {B^{\dagger} }( P_{t }^{\dagger}  - 2{D^{\dagger} }G_r^{ - 1}{{o}_{d}} )\bigg),
\end{align}
This means that the pdf of the controller follows a normal distribution, as given by:
\begin{equation}
c(u_{t-1}|x_{t-1})\sim \mathcal{N}(v_{t-1},R_{t}),
\end{equation}
where
\begin{align}
	R_t&=\bigg({\Omega ^{ - 1}} + {B^{\dagger} }( {D^{\dagger} }G_r^{ - 1}D + {M_t})B\bigg)^{-1}.
\end{align}
	\section{State space model for spin $\frac{1}{2}$}\label{spinn_1_2}
	\setcounter{equation}{0} % ap
\renewcommand{\theequation}{E.\arabic{equation}}
	The evolution of a  spin-1/2 system in interaction with an external environment  can be described by the following master equation:
	\begin{equation}\label{vonspin}
		\dfrac{d\rho(t)}{dt}=-i[H,\rho(t)]+\Theta\bigg(\sigma_{-}\rho(t)\sigma_{+}-\dfrac{1}{2}\{\sigma_{+}\sigma_{-},\rho(t)\}\bigg),
	\end{equation}
where, as given in equation (\ref{58_}), $H=\dfrac{1}{2}\sigma_3+\dfrac{1}{2}(\sigma_1+\sigma_2) {u}_{t}$ such that ${u}_{t}$ is an external electric field, and $\sigma_1,\sigma_2, \sigma_3$ are  the relevant Pauli matrices and $\sigma_{\pm}=\dfrac{\sigma_{1}\pm i\sigma_{2}}{2}$. The Hamiltonian $H$ can be written in matrix form as follows:
	\begin{equation}
		H=\dfrac{1}{2}\left(
		\begin{array}{cc}
			1 & {u}_{t}(1-i) \\
			{u}_{t}(1+i) & -1 \\
		\end{array}
		\right).
	\end{equation} 
Writing the density operator in terms of its elements, the von-Neumann equation (\ref{vonspin}) can be re-written as:
	\begin{align}
		&	\dfrac{d}{dt}\left(
		\begin{array}{cc}
			\rho_{00}(t) & \rho_{01}(t)\\
		\rho^*_{01}(t)  &\rho_{11}(t)
		\end{array}
		\right)=\dfrac{-i}{2}\left(
		\begin{array}{cc}
			{u}_{t}\big((1-i) \rho^*_{01}(t)-(1+i) \rho_{01}(t)\big) &2\rho_{01}(t)+{u}_{t}(1-i)(\rho_{11}(t)-\rho_{00}(t))\\
			-2\rho^*_{01}(t) +{u}_{t}(1+i)(\rho_{00}(t) -\rho_{11}(t) ) & {u}_{t}\big((1+i)\rho_{01}(t) -(1-i)\rho^*_{01}(t) \big)  \\
		\end{array}
		\right)\\\nonumber&-\left(	\begin{array}{cc}
			\Theta\rho_{00}(t)&\dfrac{\Theta}{2}\rho_{01}(t)\\
			 \dfrac{\Theta}{2}\rho^*_{01}(t)  &-\Theta\rho_{00}(t)
		\end{array}\right)
	\end{align}
where the elements $\rho_{00}(t), \rho_{01}(t), \rho^*_{01}(t)$ and $\rho_{11}(t)$ are the elements of the density operator $\rho(t)$. Vectorizing the above equation yields:
	\begin{align}
	&	\dfrac{d}{dt}\underbrace{\left(
		\begin{array}{cc}
			\rho_{00}(t)\\ \rho_{11}(t)\\
			\rho_{01}(t) \\ \rho^*_{01}(t)
		\end{array}
		\right)}_{x(t)}= \underbrace{\left(
		\begin{array}{cccc}
				-\Theta & 0 & 0 & 0\\
	\Theta & 0 & 0& 0\\
			0 & 0 & -i-\dfrac{\Theta}{2}& 0 \\
			0 & 0 & 0& i-\dfrac{\Theta}{2}\\
		\end{array}\right)}_{\tilde{A}} \underbrace{\left(
		\begin{array}{cc}
			\rho_{00}(t)\\ \rho_{11}(t)\\
			\rho_{01}(t) \\ \rho^*_{01}(t)
		\end{array}
		\right)}_{x(t)}\\&+i \underbrace{\dfrac{ 1}{2}\left(
		\begin{array}{cccc}
			0 & 0 & (1+i) & -(1-i)\\
			0 & 0 & -(1+i) & (1-i)\\
			(1-i) & -(1-i)&0 & 0 \\
			-(1+i) & (1+i)&0 & 0 \\
		\end{array}\right)}_{\tilde{N}} \underbrace{\left(
	\begin{array}{cc}
		\rho_{00}(t)\\ \rho_{11}(t)\\
			\rho_{01}(t) \\ \rho^*_{01}(t)
	\end{array}
		\right)}_{x(t)}{u}_{t},
	\end{align}
from which we find the state equation for the spin-1/2 system to be given by:
	\begin{equation}
	\dfrac{d{x}(t)}{d t} =(\tilde{A}+i\tilde{N} {u}_{t} ){x}(t).
	\end{equation}
%his shows the form given in (\ref{NLVN_pa}) for $\dfrac{1}{2}$-spin system.
	\section{State space model for a three level system}\label{spinn_1}
	\setcounter{equation}{0} % ap
\renewcommand{\theequation}{F.\arabic{equation}}
	The Hamiltonian describing the interaction between a $\Lambda$-type atomic  system, examined in Section (\ref{atomicsystem}), and an electric field ${u}_{t}$ can be written in matrix form as follows:
	\begin{equation}
		H=H_0+{u}_{t}H_1=\left(\begin{array}{ccc}
			\dfrac{3}{2} & 0 & 0  \\
			0 & 1 & 0  \\
			0 & 0&0   \\
		\end{array}\right)+\left(\begin{array}{ccc}
			0 & 0 & 1  \\
			0 & 0 & 1  \\
			1 & 1&0   \\
		\end{array}\right) {u}_{t}=\left(\begin{array}{ccc}
			\dfrac{3}{2} & 0 & {u}_{t}  \\
			0 & 1& {u}_{t}  \\
			{u}_{t} & {u}_{t}&0   \\
		\end{array}\right)
	\end{equation} 
By assuming that the system interacts with an external environment described by one Linblad  operator, $L_{02}=\sqrt{\Theta}\ket{0}\bra{2}$, and  expanding the density operator in terms of its matrix elements the von-Neumann equation given in (\ref{LVN_eq}) can be written as:
{\small	\begin{align}
		&	i\dfrac{d}{dt}\left(\begin{array}{ccc}
			\rho_{00}(t) & 	\rho_{01}(t) & 	\rho_{02}(t)\  \\
			\rho^*_{01}(t) & 	\rho_{11}(t) & \rho_{12}(t)\ \\
			\rho^*_{02}(t) &  \rho^*_{12}(t) &	\rho_{22}(t)
		\end{array}\right)=\\\nonumber&\left(\begin{array}{ccc}
			u_t\big(\rho^*_{02}(t) -\rho_{02}(t)\big)& \dfrac{1}{2}\rho_{01}(t)+u_t\big( \rho^*_{12}(t)-\rho_{02}(t)\big ) & \dfrac{3}{2}\rho_{02}(t) +u_t\big(\rho_{22}(t)-\rho_{00}(t) -\rho_{01}(t)\big)  \\\\
			-\dfrac{1}{2} \rho^*_{01}(t)+{u}_{t}(\rho^*_{02}(t)- \rho_{12}(t))& {u}_{t}\big( \rho_{12}^*(t)- \rho_{12}(t)\big)& \rho_{12}(t)+{u}_{t}\big(\rho_{22}(t) -\rho_{01}^*(t)-\rho_{11}(t)\big) \\\\
			-\dfrac{3}{2} \rho_{02}^*(t) +{u}_{t}\big(\rho_{00} +\rho_{01}^*(t)-\rho_{22}(t)\big)& - \rho_{12}^*(t)+{u}_{t}(\rho_{01}(t)+\rho_{11}(t)-\rho_{22} (t)) & {u}_{t}\big(\rho_{02}(t) + \rho_{12}(t)-\rho_{02}^*(t) - \rho_{12}^*(t)\big) \\
		\end{array}\right)\\\nonumber&+\dfrac{i}{2}\Theta\left(\begin{array}{ccc}
			-2\rho_{00}(t) & 	-\rho_{01}(t) & -	\rho_{02}(t) \\
			-\rho_{01}^*(t)& 	0 &  0 \\
			-\rho^*_{02}(t) &  	0 &	2\rho_{00}(t)
		\end{array}\right).
\end{align}}
which when using the vectorisation defined in Eq.(\ref{x_t_vet}), gives,
\newpage
\begin{eqnarray}
	&\dfrac{d}{dt}\underbrace{\left(
		\begin{array}{ccccccccc}
			\rho_{00} (t) \\
			\rho_{11} (t)\\
			\rho_{22}(t) \\
			\rho_{01}(t) \\
			\rho_{02}(t)  \\
			\rho_{01} ^* (t) \\
			\rho_{02}^*(t)  \\
			\rho_{12}(t)\\
			\rho_{12}^* (t) \\
		\end{array}\right)}_{x(t)}=\underbrace{\left(
		\begin{array}{ccccccccc}
			-\Theta & 0 & 0 & 0 & 0 & 0 & 0 & 0 & 0 \\
			0 & 0 & 0 & 0 & 0 & 0 & 0 & 0 & 0 \\
			\Theta & 0 & 0 & 0 & 0 & 0 & 0 & 0 & 0 \\
			0 & 0 & 0 & -\dfrac{i}{2}-\dfrac{\Theta}{2} & 0 & 0 & 0 & 0 & 0 \\
			0 & 0 & 0 & 0 & -\dfrac{3i}{2}-\dfrac{\Theta}{2} & 0 & 0 & 0 & 0 \\
			0 & 0 & 0 & 0 & 0 & \dfrac{i}{2}-\dfrac{\Theta}{2}  & 0 & 0 & 0 \\
			0 & 0 & 0 & 0 & 0 & 0 & \dfrac{3i}{2}-\dfrac{\Theta}{2}  & 0 & 0 \\
			0 & 0 & 0 & 0 & 0 & 0 & 0 & -i& 0 \\
			0 & 0 & 0 & 0 & 0 & 0 & 0 & 0 & i\\
		\end{array}
		\right)}_{\tilde{A}} \underbrace{\left(
		\begin{array}{ccccccccc}
			\rho_{00} (t) \\
			\rho_{11} (t)\\
			\rho_{22}(t) \\
			\rho_{01}(t) \\
			\rho_{02}(t)  \\
			\rho_{01} ^* (t) \\
			\rho_{02}^*(t)  \\
			\rho_{12}(t)\\
			\rho_{12}^* (t) \\
		\end{array}\right)}_{x(t)} \nonumber \\
	&+i\underbrace{\left(
		\begin{array}{ccccccccc}
			0 & 0 & 0 & 0 & 1 & 0 & -1 & 0 & 0 \\
			0 & 0 & 0 & 0 & 0 & 0 & 0 & 1 & -1 \\
			0 & 0 & 0 & 0 & -1& 0 & 1 & -1 & 1 \\
			0 & 0 & 0 & 0 & 1 & 0 & 0 & 0 & -1 \\
			1 & 0 & -1 & 1 & 0 & 0 & 0 & 0 & 0 \\
			0 & 0 & 0 & 0 & 0 & 0 & -1 & 1 & 0 \\
			-1 & 0 & 1& 0 & 0 & -1 & 0 & 0 & 0 \\
			0 & 1 & -1 & 0 & 0 & 1 & 0 & 0 & 0 \\
			0 & -1 & 1 & -1 & 0 & 0 & 0 & 0 & 0 \\
		\end{array}
		\right)}_{\tilde{N}}\underbrace{\left(
		\begin{array}{ccccccccc}
			\rho_{00} (t) \\
			\rho_{11} (t)\\
			\rho_{22}(t) \\
			\rho_{01}(t) \\
			\rho_{02}(t)  \\
			\rho_{01} ^* (t) \\
			\rho_{02}^*(t)  \\
			\rho_{12}(t)\\
			\rho_{12}^* (t) \\
		\end{array}\right)}_{x(t)} {u}_{t}.
\end{eqnarray}
Hence, we find the form of the state equation  given in Eq.(\ref{NLVN_pa}) as follows,
\begin{equation}
	\dfrac{d{x}(t)}{d t} =(\tilde{A}+i\tilde{N} u(t) ){x}(t).
\end{equation}

%\newpage{\pagestyle{empty}\cleardoublepage}
	
\end{document}